\documentclass[a4paper,UKenglish]{llncs}

\usepackage{microtype}
\usepackage{amsmath,amssymb}
\usepackage{thmtools}
\usepackage{hyperref}
\usepackage[capitalise]{cleveref}
\usepackage{apptools}
\usepackage{xpatch}
\usepackage{bm}

\declaretheorem[name=Observation]{obs}
\declaretheorem[name=Definition]{defn}

\usepackage{color}
\usepackage[noadjust]{cite}
\usepackage{xargs}                      
\usepackage[pdftex,dvipsnames]{xcolor}
\usepackage{tikz}
\usetikzlibrary{positioning}
\definecolor {processblue}{cmyk}{0.96,0,0,0}
\usepackage[colorinlistoftodos,prependcaption,textsize=small]{todonotes}
\newcommandx{\unsure}[2][1=]{\todo[linecolor=red,backgroundcolor=red!25,bordercolor=red!75,#1]{#2}}
\newcommandx{\change}[2][1=]{\todo[linecolor=blue,backgroundcolor=blue!25,bordercolor=blue,#1]{#2}}
\newcommandx{\comment}[2][1=]{\todo[linecolor=OliveGreen,backgroundcolor=OliveGreen!25,bordercolor=OliveGreen,#1]{#2}}
\newcommandx{\improvement}[2][1=]{\todo[linecolor=Plum,backgroundcolor=Plum!25,bordercolor=Plum,#1]{#2}}

\newcommandx{\mat}[1]{\textcolor{red}{#1}}
\newcommandx{\rian}[1]{\textcolor{blue}{#1}}

\usepackage{algorithm2e}
\usepackage{ifthen}

\title{Recognizing $k$-Clique Extendible Orderings}

\author{Mathew Francis\inst{1} \and Rian Neogi\inst{2} \and Venkatesh Raman\inst{2}}
\institute{Indian Statistical Institute, Chennai Centre, India\\\email{mathew@isichennai.res.in} \and
The Institute of Mathematical Sciences, HBNI, Chennai, India \\\email{rianneogi@gmail.com, vraman@imsc.res.in}}

\newcommand{\BB}{{\mathcal B}}
\newcommand{\nn}{{\mathbb N}}

\newcommand{\ProbVerify}{\textsc{Verify $k$-C-E Ordering}}

\newcommand{\ProbFind}{\textsc{Find $k$-C-E Ordering}}

\newcommand{\defproblem}[3]{
  \vspace{3mm}
\noindent\fbox{
  \begin{minipage}{.95\textwidth}
  \begin{tabular*}{\textwidth}{@{\extracolsep{\fill}}lr} \textsc{#1}\\ \end{tabular*}
  {\bf{Input:}} #2  \\
  {\bf{Question:}} #3
  \end{minipage}
  }
  \vspace{2mm}
}

\renewenvironment{proof}{\textit{Proof.}}{\qed}

\begin{document}
\pagestyle{plain}
\maketitle

%\comment{Find better ref for treewidth}

\begin{abstract}
We consider the complexity of recognizing $k$-clique-extendible graphs ($k$-C-E graphs) introduced by Spinrad (Efficient Graph Representations, AMS 2003), which are generalizations of comparability graphs.
A graph is $k$-clique-extendible if there is an ordering of the vertices such that whenever two $k$-sized overlapping cliques 
$A$ and $B$ have $k-1$ common vertices, and these common vertices appear between the two vertices $a,b\in (A\setminus B)\cup (B\setminus A)$ in the ordering, there is an edge between $a$ and $b$, implying that $A\cup B$ is a $(k+1)$-sized clique.  Such an ordering is said to be a $k$-C-E ordering.
These graphs arise in applications related to modelling preference relations.
Recently, it has been shown that a maximum sized clique in such a graph can be found in $n^{O(k)}$ time [Hamburger et al. 2017] when the ordering is given. 
When $k$ is $2$, such graphs are precisely the well-known class of comparability graphs and when $k$ is $3$ they are called triangle-extendible graphs. It has been shown that triangle-extendible graphs appear as induced subgraphs of visibility graphs of simple polygons, and the complexity of recognizing them has been mentioned as an open problem in the literature.

While comparability graphs (i.e. $2$-C-E graphs) can be recognized in polynomial time, we show that
recognizing $k$-C-E graphs is NP-hard for any fixed $k \geq 3$ and {\sc co-NP}-hard when $k$ is part of the input. While our NP-hardness reduction for $k \geq 4$ is from the betweenness problem, for $k=3$, our reduction is an intricate one from the $3$-colouring problem.
We also show that the problems of determining whether a given ordering of the vertices of a graph is a $k$-C-E ordering, and
that of finding an $\ell$-sized (or maximum sized) clique in a $k$-C-E graph, given a $k$-C-E ordering,
are complete for the parameterized complexity classes {\sc co-W[1]} and {\sc W[1]} respectively, when parameterized by $k$. 
However we show that the former is fixed-parameter tractable when parameterized by the treewidth of the graph.

\end{abstract}

\section{Introduction and Motivation}

%\textbf{\textcolor{red}{TODO: add diagrams to proofs}}

%A poset $P = (V, <)$ is called a semiorder if for some utility function $\alpha: V \rightarrow R$ we have $u <_P v$ if and only if $\alpha(v) - \alpha(u) > 1$. Semiorder were introduced as a possible mathematical model of a preference relation. Semiorders were designed to model imprecision in the utility function. We may be indifferent between two elements if the difference between their utility values is small. There has been a great deal of literature on semiorders and preference; see \cite{interval_book} and \cite{semiorders_book}.

%Semiorders are a transitive relation. However, a variety of evidence suggests that preferences are not always transitive. This has led to a lot of discussion over the meaning of preference, see \cite{fishburn}. To resolve this issue, Hamburger et. al.\cite{doublethreshold} introduce a notion called Double Threshold Digraphs. In Double Thresold Digraphs, there are two threshold values $t_1$ and $t_2$ such that, for two elements $u$ and $v$, if $\alpha(v) - \alpha(u) < t_1$, $(u,v)$ is not a preference. If $t_1 \leq \alpha(v) - \alpha(u) \leq t_2$, $(u,v)$ can either be a preference or not, and if $\alpha(v) - \alpha(u) < t_2$, $v$ is prefered over $u$. The ratio $t_2/t_1$ is denoted as $\lambda$. They prove that when $\lambda < 2$, $\alpha$ is a transitive relation. Hence, the $\lambda$ parameter can be thought of as a generalization of transitivity.

An undirected graph is a comparability (or transitively orientable) graph if the edges can be oriented in a way that for any three vertices $u, v, w$ whenever there is a (directed) edge from $u$ to $v$ and an edge from $v$ to $w$, there is an edge from $u$ to $w$.
% Equivalently, an undirected graph is a comparability graph if there is an ordering of the vertices such that whenever there are edges $(u,v)$ and $(v,w)$ with $v$ appearing between $u$ and $w$ in the ordering, there is an edge between $u$ and $w$.
They are a well-studied class of graphs~\cite{golumbic, mohring} and they can be recognized in polynomial time~\cite{lin_trans}.
Spinrad~\cite{spinrad} generalized this class of graphs and introduced the notion of $k$-clique-extendible orderings (abbr. $k$-C-E ordering) on the vertices of a graph defined as follows.

\begin{defn}[$k$-C-E ordering, Spinrad~\cite{spinrad}]\label{def:kceordering}
	An  ordering $\phi$ of the vertices of a graph $G= (V,E)$ is a $k$-clique-extendible ordering (or $k$-C-E ordering) of $G$ if, whenever $X$ and $Y$ are two overlapping cliques of size $k$ such that $|X \cap Y| = k-1$, $X \setminus Y=\{a\}$, $Y \setminus X=\{b\}$, and all the vertices in $X\cap Y$ occur between $a$ and $b$ in $\phi$, we have $(a,b)\in E(G)$ and hence $X \cup Y$ is a $(k+1)$-clique.
\end{defn}

A graph $G$ is said to be \emph{$k$-clique-extendible} (\emph{$k$-C-E} for short) if there exists a $k$-clique-extendible ordering $\phi$ of $G$. 
It can be observed that comparability graphs are exactly the $2$-clique-extendible graphs. Spinrad~\cite{spinrad} observed that $3$-clique-extendible graphs, also called triangle-extendible graphs, arise in the visibility graphs of simple polygons and that a maximum clique can be found in polynomial time in such graphs if a $3$-clique-extendible ordering is given. This result has been generalized to obtain an $n^{O(k)}$ algorithm for finding a maximum clique in $k$-C-E graphs (given with a $k$-C-E ordering) on $n$ vertices~\cite{doublethreshold}. 
%The questions of whether there are polynomial time algorithms to recognise $3$-C-E graphs or to find a maximum clique in such graphs given just the adjacency matrix, have been mentioned as open problems~\cite{spinrad}.
The question of whether there is a polynomial time algorithm to recognise $3$-C-E graphs has been mentioned as an open problem~\cite{spinrad}.

We believe that $k$-C-E graphs are natural generalizations of comparability graphs and our main contribution in this paper is a serious study of this class of graphs. Our results show that recognizing $k$-C-E graphs is NP-hard for any fixed $k \geq 3$ and also {\sc co-NP}-hard when $k$ is part of the input. This solves the open problem regarding the complexity of recognizing $3$-C-E graphs and we hope that our results will trigger further study of $k$-C-E graphs in general.

If an ordering of the vertices is given, then it is easy to get an $n^{O(k)}$ algorithm to determine whether it is a $k$-C-E ordering of the graph (see \cref{section:verify}). We show that this problem is {\sc co-NP}-complete and also complete for the parameterized complexity class {\sc co-W[1]}. The reduction also implies that unless the Exponential Time Hypothesis fails, this problem does not have an $f(k)n^{o(k)}$ algorithm for any function $f$ of $k$. However, we show that the problem is fixed-parameter tractable when parameterized by the treewidth of the graph, that is, there is an $f(tw) n^{O(1)}$ algorithm for the problem, where $tw$ is the treewidth of the graph (see \cref{section:prelims} for definitions).

%is a graph in which it is possible to assign orientations to each of its edges such that the resulting directed graph is acyclic. A visibility graph is a graph whose vertices are the vertices of a simple polygon, where $(u,v)$ is an edge if $u$ is visible from $v$ in the polygon.
%When $k=2$, $k$-clique orderable graphs are the same as a comparibility graphs.  Graphs which are $3$-clique orderable arise in visibility graphs. 
%Spinrad gives a polynomial time algorithm for finding a maximum clique in a $3$-CEO graph when given the $3$-clique extendable ordering on the vertices \cite{spinrad}. Recently, this algorithm was formalized and extended to work for any $k$-CEO graph, where they find a maximum clique in time $O(kn^{O(k)})$ provided they have the $k$-clique extendable ordering on the vertices \cite{doublethreshold}.  

%In \cite{doublethreshold}, they proved that any graph $G$ with a small value of $\lambda$ also has a small value $k$ such that there exists a $k$-clique extendable ordering of $G$. In particular $k \leq \lambda+1$. They also provide an algorithm to compute the maximum clique in $O(km^{O(k)})$ time, and by extension in $O(\lambda m^{O(\lambda)})$ time. 
%In this paper, we extend the results on $k$-clique extendable orderings. We prove that finding a $k$-CEO of a graph is Para-NP-hard, when parameterized by $k$ and that checking whether a given ordering is a $k$-CEO of a graph is W[1]-Complete when parameterized by $k$.
%We also give an algorithm to find a $k$-CEO of a graph with bounded treewidth.
\medskip

\noindent
{\bf Organization of the paper.} In the next section, we give the necessary notation and definitions.
% including those of treewidth and fixed-parameter tractability.
In \cref{section:basic}, we prove some results about $k$-C-E graphs which are used in our reductions in later sections. In \cref{section:verify}, we show that the problem of checking whether a given ordering is a $k$-C-E ordering is {\sc co-NP}-complete and {\sc co-W[1]}-complete. In this section, we also show that the problem is fixed-parameter tractable when parameterized by the treewidth of the graph. In \cref{section:clique_in_kceo} we show that the $n^{O(k)}$ algorithm for finding maximum clique in a $k$-C-E graph~\cite{doublethreshold} is likely optimal. \cref{section:ordering} gives our main NP-hardness reductions for the problem of recognizing $k$-C-E graphs. We give two reductions, one for $k=3$ and another for $k\geq 4$. We list some open problems in \cref{section:conclusion}.

\section{Preliminaries} \label{section:prelims}

\begin{defn}[Fixed-Parameter Tractability]
\label{defn:fpt}
A parameterized problem (or a language) $L \subseteq \Sigma^* \times \nn$ is said to be fixed-parameter tractable (FPT) if there exists an algorithm $\BB$, a constant $c$ and a computable function $f: \nn \times \nn$ such that given any $(I,k) \in \Sigma^* \times \nn$, $\BB$ runs in at most $f(k)\cdot |I|^c$ time and decides correctly whether $(I,k) \in L$ or not. Here $f$ is a function only of $k$, and $c$ is a constant independent of $k$.
We call algorithm $\BB$ as fixed-parameter algorithm, and we also denote a runtime like $f(k) |I|^c$, a FPT runtime. 
FPT also denotes the class of fixed-parameter tractable problems.
Here $|I|$ is the size of the input and $k$ is the parameter. 
\end{defn}
	
There is also an intractability (hardness) theory in parameterized complexity that is characterized by a hierarchy of complexity classes $FPT \subseteq W[1] \subseteq W[2] \subseteq \cdots \subseteq XP$. It is believed that the containments are strict and there are canonical complete problems under parameterized reductions. 
	
\begin{defn}[{Parameterized Reduction, $W[1]$, $W[1]$-complete}]
There is a parameterized reduction from a parameterized problem $P_1$ to a parameterized problem $P_2$, if every instance $(x,k)$ of $P_1$ can be transformed in FPT time to an equivalent instance $(x', k')$ where $k'$ is just a function of $k$.
The {\sc Clique} problem that asks whether a given undirected graph has a clique of size $k$ is a canonical $W[1]$-complete problem, where $k$ is the parameter.  Parameterized problems that have a parameterized reduction to the {\sc Clique} problem form the class $W[1]$.
\end{defn}
We refer readers to the recent textbook~\cite{pa_book} for further discussions on parameterized complexity.
	
\begin{defn}[Tree-decomposition and treewidth~\cite{treewidth}]
A \emph{tree decomposition} of a graph $G$ is a pair $\mathcal{T} = (T, \{X_t \}_{t \in V(T)})$, where $T$ is a tree whose every node $t$ is assigned a vertex subset $X_t \subseteq V(G)$, called a bag, such that the following three conditions hold : $(i)\bigcup_{t \in V(T)} X_t = V(G)$. $(ii)$ For every $uv \in E(G)$, there exists a node $t$ of $T$ such that bag $X_t$ contains both $u$ and $v$. $(iii)$ For every $u \in V(G)$, the set $T_u = \{t \in V(T) \mid\ u \in X_t\}$ induces a connected subtree of $T$.
The width of tree decomposition $\mathcal{T} = (T, \{X_t \}_{t \in V(T)})$ equals $\max_{t \in V(T)} \{ |X_t| -1 \}$. The treewidth of a graph $G$, denoted by $tw(G)$, is the minimum possible width of a tree decomposition of $G$.
\end{defn}

The following conjecture, known as the Exponential Time Hypothesis, is used to provide lower bounds for hard problems.
\smallskip

\noindent
{\textbf{Exponential Time Hypothesis} (ETH)~\cite{eth_defn}: There is a positive real $s$ such that $3$-CNF-SAT cannot be solved in time $2^{sn}(n+m)^{O(1)}$ where $n$ is the number of variables, and $m$ is the number of clauses in the formula.}
\medskip

See also~\cite{eth_lowerbounds} for a survey of various lower bound results using ETH.
\medskip

All graphs considered in this paper are undirected and simple. Given a graph $G$, by $V(G)$ we denote the set of vertices in the graph and by $E(G)$ we denote the set of edges in the graph.
% \change{Removed the definition of $\overline{E}(G)$}
Let $G$ be a graph. For a subset of vertices $S \subseteq V(G)$, we define $G[S]$ as the induced subgraph of $G$ having vertex set $S$. 

Given a linear order $\phi$ of a set $A$, we write $a<_{\phi} b$ to mean that $a$ and $b$ are two elements of $A$ such that $a$ occurs before $b$ in $\phi$. Also, we write $\phi=(a_1,a_2,\ldots,a_n)$ to mean that $A=\{a_1,a_2,\ldots,a_n\}$ and $a_1<_{\phi} a_2<_{\phi}\cdots<_{\phi} a_n$.
% By $\phi(k)$, we denote the $k$th element in the ordering, i.e. $a_k$. \comment{The definition of $\phi(k)$ is needed only for the treewidth proof, it seems. Is it really required for that proof?} 
We say that a vertex $b$ comes between vertices $a$ and $c$ in $\phi$ if $a <_{\phi} b <_{\phi} c$ or $c <_{\phi} b <_{\phi} a$. By $\phi^{-1}$ we denote the reverse of $\phi$, that is, $a <_{\phi^{-1}} b$ if and only if $b <_{\phi} a$.

Given an ordering $\phi$ of a set $V$ and a set $S \subseteq V$, we define $\phi|_S$ to be the ordering of the elements of $S$ in the order in which they occur in $\phi$. Further, we say that $a,b \in S$ are the endpoints of $S$ if $a$ is the first element of $\phi|_S$ and $b$ is the last element of $\phi|_S$.
%Given an ordering $\phi$ of a set $V$ and sets $S_1,S_2\subseteq V$, we say that $S_1<_{\phi} S_2$ if every element of $S_1$ occurs before every element of $S_2$ in $\phi$; i.e. for every $u\in S_1$ and $v\in S_2$, we have $u<_{\phi} v$.
Given two disjoint sets $A$ and $B$, and orderings $\phi _1  = (a_1,a_2,\ldots,a_n)$ of the set $A$ and $\phi _2  = (b_1,b_2,\ldots,b_m)$ of the set $B$, we define $\phi _1 + \phi _2 = (a_1,a_2,\ldots,a_n,b_1,b_2,\ldots,b_m)$ that is an ordering on the set $A \cup B$, that is, $+$ is the concatenation operator on orderings. 
We will abuse notation to allow sets to be used with the concatenation operator: if $\gamma$ is an expression that is a concatenation of orderings and sets, we say that an ordering $\phi$ is of the form $\gamma$, if there exists an ordering for each set appearing in $\gamma$ such that replacing each set with its corresponding ordering in $\gamma$ yields the ordering $\phi$.

A clique in a graph is a set of vertices that are pairwise adjacent in the graph. An independent set is a set of vertices that are pairwise non-adjacent. %\change{Removed definition of $P_3$ and its middle vertex}
%A separator of a graph $G$ is a set of vertices $S \subseteq V(G)$ such that $G[V(G) \setminus S]$ has more connected components than $G$. 
Given subsets $S,A,B \subseteq V(G)$, we say that $S$ separates $A$ and $B$ if there is no path from $A$ to $B$ in $G[V(G) \setminus S]$.
For a pair $u,v$ of nonadjacent vertices of a graph, by identifying $u$ with $v$, we mean adding the edges $(u,w)$ for all $w \in N(v)\setminus N(u)$ and then deleting $v$. %Intuitively, this means that we identify $u$ and $v$ as the same vertex.
%We do not add any self-loops or multiedges in the process of identification.

We denote by $K^-_n$ the graph obtained by removing an edge from the complete graph $K_n$ on $n$ vertices. Given an ordering $\phi$ of the vertices of a graph $G$, we say that an induced subgraph $H$ of $G$ is an \emph{ordered $K^-_t$} in $\phi$ if $\phi|_{V(H)}=(h_1,h_2,\ldots,h_t)$ and $E(H)=\{(h_i,h_j)\mid 1\leq i<j\leq t\}\setminus \{(h_1,h_t)\}$. 
It follows that an ordering of the vertices of a graph is a $k$-C-E ordering if and only if it contains no ordered $K^-_{k+1}$.

\section{Basic Results} \label{section:basic}

We start with the following observations which are used throughout the paper.
%The following are easy to verify.\comment{You can just list all these in one or two theorems without proofs? I didn't do it as you are referring to these, so you need to appropriately label them.}
%\begin{obs} \label{conflictingclique}
%An ordering is a $k$-CEO if and only if it has no conflicting pair of cliques (or a conflicting pair of vertices).
%\end{obs}
%\begin{proof}
%Follows by definition. 
%\end{proof}

\begin{obs} \label{obs:reverse}
An ordering $\phi = \{v_1,v_2,\ldots,v_n\}$ is a $k$-C-E ordering, if and only if its reverse ordering, $\phi^{-1} = \{v_n,v_{n-1},\ldots,v_1\}$ is also a $k$-C-E ordering.
\end{obs}
%\begin{proof}
%Suppose not, then $\sigma^{-1}$ has conflicting cliques, $A, B$. Note that $A, B$ would also be conflicting cliques in $\sigma$, so $\sigma$ is not a $k$-CEO, a contradiction. Similar argument holds for the reverse direction.
%\end{proof}

\begin{obs} \label{obs:hereditary}
Given a graph $G$ and an induced subgraph $H$ of $G$, if an ordering $\phi $ is a $k$-C-E ordering of $G$, then $\phi|_{V(H)}$ is a $k$-C-E ordering of $H$. Thus every induced subgraph of a $k$-clique-extendible graph is also $k$-clique-extendible.
\end{obs}

\begin{obs}\label{obs:coloring}
If $G$ is a $k$-colourable graph with colour classes $V_1,\ldots,V_k$, then any ordering $\phi$ of $V(G)$ of the form $V_1 + V_2 +\cdots+ V_k$ is a $k$-C-E ordering of $G$. Thus, every $k$-colourable graph is $k$-clique-extendible.
\end{obs}
%\begin{proof}
%Suppose not, then there is a conflicting pair of vertices in $\phi[V(H)]$, and the same pair will be a conflicting pair in $\phi$ also.
%\end{proof}
%
% \begin{obs} \label{hereditary}
% (Hereditary Property) If $G$ is a $k$-clique orderable graph and $H$ is a subgraph of $G$, then $H$ is also $k$-clique extendible.
% \end{obs}
%\begin{proof}
%Let $\phi$ be a $k$-CEO of $G$. Take the ordering of $\phi$ induced on $V(H)$. This is a $k$-CEO of $H$.
%\end{proof}

It is not difficult to see that any $k$-clique-extendible ordering of a graph is also a $(k+1)$-clique-extendible ordering of it. Thus, every $k$-clique-extendible graph is also a $(k+1)$-clique-extendible graph. Note that every graph on $n$ vertices is trivially $n$-clique-extendible. So the notion of $k$-clique-extendibility gives rise to a hierarchy of graph classes starting with comparability graphs and ending with the entire set of graphs. This motivates the use of $k$ as a graph parameter.

%\textbf{\textcolor{red}{TODO: mention preservation of ordering}}

%\subsection{Small Separators}

We prove a lemma that will help us construct a $k$-C-E ordering of a graph from $k$-C-E orderings of its subgraphs.

\begin{lemma}\label{lem:separator2}
For a graph $G$, let $V_1, V_2 \subseteq V(G)$ and let $\sigma _1, \sigma _2$ be $k$-C-E orderings of $G[V_1]$ and $G[V_2]$ respectively for any $k\ge 2$, such that the following hold
\begin{enumerate}
	\item  $V_1 \cup V_2 = V(G)$
	\item  $V_1 \cap V_2$ separates $G$ into components $V_1\setminus V_2$ and $V_2\setminus V_1$
	\item  $\sigma_1|_{V_1 \cap V_2} = \sigma_2|_{V_1 \cap V_2}$
	\item if $C$ is a $(k-1)$-clique in $V_1 \cap V_2$ and $u,v$ are the endpoints of $C$ in $\sigma_1$, then every vertex $a \in V_1\setminus V_2$ that is adjacent to all of $C$ satisfies $u <_{\sigma _1} a <_{\sigma _1} v$
\end{enumerate}
	%(1.), (2.) $V_1 \cap V_2$ separates $V_1$ and $V_2$, (3.) $\sigma_1|_{V_1 \cap V_2} = \sigma_2|_{V_1 \cap V_2}$ and (4.) if $C$ is a $(k-1)$-clique in $V_1 \cap V_2$ and $u,v$ are the endpoints of $C$ in $\sigma_1$, then every vertex $a \in V_1\setminus V_2$ that is adjacent to all of $C$ satisfies $u <_{\sigma _1} a <_{\sigma _1} v$. 
	Then $G$ has a $k$-C-E ordering $\phi$ such that $\phi|_{V_1} = \sigma_1$ and $\phi|_{V_2} = \sigma_2$.
\end{lemma}
\begin{proof}
Let $\sigma_1|_{V_1 \cap V_2} = \sigma_2|_{V_1 \cap V_2} = (s_1, s_2, \ldots, s_p)$. Let $A_i$ be the induced ordering between $s_i$ and $s_{i+1}$ in $\sigma_1$ so we can rewrite $\sigma_1$ as 
\[\sigma_1 = A_0+s_1+A_1+s_2+A_2+\ldots + A_{p-1}+s_p+A_p\]
Similarly, let $B_i$ be the induced ordering between $s_i$ and $s_{i+1}$ in $\sigma_2$ so that 
\[\sigma_2 = B_0+s_1+B_1+s_2+B_2+\ldots +B_{p-1}+ s_p+B_p\] 
Consider the following ordering $\phi$ of $V_1 \cup V_2$.
\[\phi = A_0+B_0+s_1+A_1+B_1+s_2+A_2+B_2+\ldots + A_{p-1}+B_{p-1}+s_p+A_p+B_p\] 

That is, we `interleave' each $A_i$ and $B_i$ between the corresponding $s_i$ and $s_{i+1}$.
%	Notice that $V_1\setminus V_2=\bigcup_{i=0}^p A_i$ and $V_2\setminus V_1=\bigcup_{i=0}^p B_i$.
The ordering $\phi$ is constructed such that it preserves the internal ordering of $\sigma_1$ in $V_1$ and $\sigma_2$ in $V_2$, that is, $\phi|_{V_1} = \sigma_1$ and $\phi|_{V_2} = \sigma_2$ and thus also $\phi|_{V_1\cap V_2} = \sigma_1|_{V_1 \cap V_2} = \sigma_2|_{V_1 \cap V_2}$.
We will prove that $\phi$ is a $k$-C-E ordering of $G$. Suppose not. Then there exists a set $Q\subseteq V(G)$ that forms an ordered $K^-_{k+1}$ in $\phi$. Let $a,b$ be the endpoints of $Q$ in $\phi$.
It can't be the case that $Q \subseteq V_1$, otherwise since $\phi|_{V_1} = \sigma_1$, $Q$ would be an ordered $K_{k+1}^-$ in $\sigma_1$, contradicting the fact that $\sigma_1$ is a $k$-C-E ordering of $G[V_1]$. Similarly, it can't be the case that $Q \subseteq V_2$.
So, $Q \cap (V_1 \setminus V_2) \neq \emptyset$ and $Q \cap (V_2 \setminus V_1) \neq \emptyset$. As $V_1\cap V_2$ separates $V_1$ and $V_2$, no vertex in $V_1\setminus V_2$ is adjacent to any vertex $V_2\setminus V_1$. Since the only two vertices in $Q$ that do not have an edge between them are $a$ and $b$, we can assume without loss of generality that $a\in V_1\setminus V_2$ and $b\in V_2\setminus V_1$, and we further get that $Q\setminus\{a,b\}\subseteq V_1\cap V_2$. Since $Q\setminus\{a,b\}$ is a $(k-1)$ sized clique and $a$ is adjacent to all the vertices of $Q\setminus\{a,b\}$, by the last condition in the lemma, it must be the case that $a$ lies between the two endpoints of $Q\setminus\{a,b\}$ in $\phi$, contradicting the fact that $a$ is an endpoint of $Q$ in $\phi$.
\end{proof}

% \begin{enumerate}
% \item $V_1 \cup V_2 = V(G)$
% \item $V_1 \cap V_2$ separates $V_1$ and $V_2$
% \item $\left.\sigma _1\right|_{V_1 \cap V_2} = \left.\sigma_2\right|_{V_1 \cap V_2}$
% \item if $C$ is a $(k-1)$-clique in $V_1 \cap V_2$ and $u,v$ are the endpoints of $C$ in $\sigma _1$, then every vertex $a \in V_1\setminus V_2$ that is adjacent to all of $C$ satisfies $u <_{\sigma _1} a <_{\sigma _1} v$. 
% % \item $|V_1 \cap V_2| \leq k-1$
% % \item The first and last elements of $\sigma _1$ are in $V_1 \cap V_2$
% \end{enumerate}
% Then $G$ has a $k$-C-E ordering $\phi $ such that $\left.\phi\right|_{V_1} = \sigma _1$ and $\left.\phi\right|_{V_2} = \sigma_2$.
% \end{restatable}

\bigskip

\noindent
\textbf{Forbidden Subgraph.} We construct a forbidden subgraph for the class of $k$-clique-extendible graphs which is used to build gadgets in our NP-hard reductions.

For a positive integer $k$, let $K = \{v_1,v_2,\ldots,v_{2k-1}\}$ be a $(2k-1)$ sized clique. For every pair of vertices $v_i$ and $v_j$ in $K$, add a vertex $u_{i,j}$ such that $u_{i,j}$ is adjacent to every vertex in $K$ except $v_i$ and $v_j$. Let $I = \{u_{i,j} \mid i,j \in [2k-1], i<j\}$ be the set of all such $u_{i,j}$ for every pair of vertices in $K$. Let $F_k$ be the graph thus obtained having vertex set $K \cup I$.
See Figure~\ref{figure_fk} for an example that demonstrates the adjacencies between $I$ and $K$ when $k=3$. 
%Only three of the $10$ vertices in $I$ are shown. 
%Edges within $K$ are not shown.

%\usepackage{ifthen}
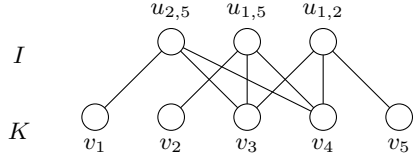
\begin{figure}[t!]
\begin{center}
\begin {tikzpicture}%[-latex ,auto ,node distance =4 cm and 5cm ,on grid ,
%semithick ,
%state/.style ={ circle ,top color =white , bottom color = white ,
%draw,black , text=black , minimum width =1 cm}]
\begin{scope} [vertex style/.style={draw,
                                       circle,
                                       minimum size=3mm,
                                       inner sep=0pt,
                                       outer sep=0pt,
										text depth=10pt
}] 
	  \path	\foreach \i in {1,...,5} 
					{
						(\i,0) coordinate[vertex style, label=below:$v_\i$](v\i)
					}
					(2,1) coordinate[vertex style, label=above:$u_{2,5}$](u25)
					(3,1) coordinate[vertex style, label=above:$u_{1,5}$](u15)
					(4,1) coordinate[vertex style, label=above:$u_{1,2}$](u12)
					(0,-0.4) coordinate[label=$K$](K)
					(0,0.6) coordinate[label=$I$](I)
	;
    \end{scope}

     \begin{scope} [edge style/.style={draw=black}]
	\foreach \i in {3,...,5}
	{
			\draw[edge style] (v\i) -- (u12);
	}
	\foreach \i in {2,...,4}
	{
			\draw[edge style] (v\i) -- (u15);
	}
	\draw[edge style] (v1) -- (u25);
	\draw[edge style] (v3) -- (u25);
	\draw[edge style] (v4) -- (u25);
     \end{scope}
\end{tikzpicture}
	\caption{Diagram depicting $F_3$. Edges in the clique are not shown, and only $3$ of the $u_{i,j}$ vertices are shown to avoid visual clutter.}
\label{figure_fk}
\end{center}
\end{figure}

\begin{lemma}\label{prop:forbidden}
$F_k$ is not $k$-clique-extendible.
\end{lemma}
\begin{proof}
%Proof by contradiction. 
% Let $n$ be the number of vertices in $F_k$. 
Suppose not. Let $\phi$ % = (v_{\phi(1)}, v_{\phi(2)},$ $...,v_{\phi(n)})$
be a $k$-C-E ordering of $F_k$. Let $\phi|_K = (v_{\sigma(1)}, v_{\sigma(2)},\ldots,$ $v_{\sigma(2k-1)})$, where $\sigma$ is a permutation on $\{1,2,\ldots,2k-1\}$.
If the vertex $u_{\sigma(1),\sigma(2k-1)}$ comes after $v_{\sigma(k)}$ in $\phi$, then the vertices in $\{v_{\sigma(1)}, v_{\sigma(2)},\ldots,v_{\sigma(k)},u_{\sigma(1),\sigma(2k-1)}\}$ form an ordered $K_{k+1}^-$. On the other hand, if $u_{\sigma(1),\sigma(2k-1)}$ comes before $v_{\sigma(k)}$ in $\phi$, then the vertices in $\{u_{\sigma(1),\sigma(2k-1)}, v_{\sigma(k)}, v_{\sigma(k+1)},\ldots,v_{\sigma(2k-2)},v_{\sigma(2k-1)}\}$ form an ordered $K_{k+1}^-$.
In both cases, we get an ordered $K_{k+1}^-$, so $\phi$ cannot be a $k$-C-E ordering, which contradicts our assumption.
\end{proof}

\section{Verifying a $k$-C-E Ordering} \label{section:verify}

In this section, we prove that even verifying whether an ordering is a $k$-clique-extendible ordering is hard (assuming $k$ is considered as part of the input, rather than a constant).

% \begin{prob} \label{prob:verify}
% Given a graph $G = (V,E)$ and an ordering of the vertices of the graph $\phi = (v_1,v_2,...,v_n)$. Check whether $\phi$ is a $k$-C-E ordering.
% \end{prob}

\defproblem{Verify $k$-C-E Ordering}{Graph $G$, integer $k$ and an ordering $\phi$ of $V(G)$}{Is $\phi$ a $k$-C-E ordering of $G$?}

{\ProbVerify} has a simple $n^{O(k)}$ algorithm as one can enumerate all ${n \choose {k+1}}$ subgraphs isomorphic to $K_{k+1}^-$, and check if any of them are ordered with respect to the ordering.  
%If we find such a pair of vertices, then we have identified a pair of conflicting cliques, in which case we return false and return true otherwise.
We prove that the problem is {\sc co-W[1]}-complete and {\sc co-NP}-complete by a reduction from and to the {\sc Clique} problem, and that the problem also cannot have a $f(k)n^{o(k)}$ algorithm assuming ETH (see \cref{section:prelims} for a definition). 
The reduction maps the YES instances of {\ProbVerify} to the NO instances of {\sc Clique} and vice-versa. 
% Hence it is a many-one reduction to {\sc coClique} (that is, the complement of the {\sc Clique} problem). 
Hence showing that {\ProbVerify} is {\sc co-NP}-complete. 

% \begin{restatable}[$\star$]{lemma}{verifyhard} \label{lem:verify_hard} 
% \label{redfromclique}
% 	{\ProbVerify} is {\sc W[1]}-hard, {\sc CoNP}-hard and there is no $f(k)n^{o(k)}$ algorithm for it unless ETH fails.
% \end{restatable}
%
% The $W[1]$-hardness of {\ProbVerify} tells us that a $f(k)n^{O(1)}$ algorithm is unlikely, but the above reduction also provides us with a ETH based lower bound for {\ProbVerify}.
%
% \begin{lemma}~\cite{clique_eth_lb1,clique_eth_lb2}
% 	Assuming ETH, there is no $f(k)n^{o(k)}$ algorithm for {\sc Clique}.
% \end{lemma}
%
% The above reduction preserves the parameter within a constant factor, so as a corollary we get the following lower bound.
%
% \begin{corollary}
% 	Assuming ETH, there is no $f(k)n^{o(k)}$ algorithm for {\ProbVerify}.
% \end{corollary}

%\begin{defn}
%For an ordering $\phi = (v_1,v_2,...,v_n)$ of vertices of $G$, and two vertices $v_i$ and $v_j$ where $i < j$, the neighbourhood-ordering-induced graph for $v_i$ and $v_j$ is the induced graph on the set $\{v_{i+1}, v_{i+2},...,v_{j-1}\} \cap N(v_i) \cap N(v_j)$ where $N(v)$ is the neighbourhood of $v$ in $G$.
%\end{defn}
% By giving a reduction from the problem to the {\sc Clique} problem, we show the following.
\begin{theorem}\label{lem:verify_hard}
{\ProbVerify} is {\sc co-W[1]}-complete, {\sc co-NP}-complete and there is no $f(k)n^{o(k)}$ algorithm for it unless {\sc ETH} fails.
\end{theorem}
\begin{proof}
We will prove hardness first by giving a reduction from {\sc Clique}. In the {\sc Clique} Problem, we are given a graph $G$ and a positive integer $k$ and asked to check whether there exists a clique of size $k$ in $G$.
	
Given $G$, let $\phi = (v_1,v_2,\ldots,v_n)$ be an arbitrary ordering of its vertices. We construct $G' = (V',E')$, where $V' = V \cup \{a,b\}$ and $E' =  E \cup \{(a,v),(b,v) \mid v \in V\}$. We then ask whether the ordering $\sigma = (a,v_1,v_2,\ldots,v_n,b)$ is a $(k+1)$-C-E ordering of $G'$. We claim that $G$ has a $k$-clique if and only if $\sigma$ is not a $(k+1)$-C-E ordering of $G'$.
	
%\begin{prop}
%If $\sigma$ is a $(k+1)$-CEO of $G'$ then $G$ does not have a $k$-clique.
%\end{prop}
%\begin{proof}
% Proof by contradiction. 
	
Suppose $G$ has a $k$-clique. Without loss of generality, let the clique be $\{v_1,v_2,\ldots,v_k\}$. Then the vertices in $\{a,v_1,v_2,\ldots,v_k,b\}$ form an ordered $K_{k+2}^-$. Thus we conclude that $\sigma$ is not a $(k+1)$-C-E ordering. 
Conversely, if $G$ does not have a $k$-clique, then $G'$ cannot have a $(k+1)$-clique, so any ordering is trivially a $(k+1)$-C-E ordering of $G'$. %In particular, $\sigma$ is a $(k+1)$-C-E ordering of $G'$.

The above reduction proves that the problem is {\sc co-W[1]}-hard and {\sc co-NP}-hard. It remains to show that the problem is in {\sc co-W[1]} and in {\sc co-NP}. For this, we give a reduction to the {\sc Clique} problem.
	
Given the ordering $\sigma$ of vertices of $G$, we do the following.
For every pair of non-adjacent vertices $(u,v)$, let $V_{u,v}$ be the set of common neighbours of both $u$ and $v$, that appear between $u$ and $v$ in the ordering.
That is, $w \in V_{u,v}$ if and only if $(w,u),(w,v) \in E(G)$ and $u <_{\sigma} w <_{\sigma} v$ or $v <_{\sigma} w <_{\sigma} u$.
Let $G_{u,v}$ be the induced subgraph $G[V_{u,v}]$. Let $\overline E(G)$ denote the pairs of vertices in $G$ which are non-adjacent. We define $G' = \dot \bigcup_{(u,v) \in \overline E(G)} G_{u,v}$. That is, $G'$ is the disjoint union of $G[V_{u,v}]$ for all pairs $(u,v)$ that are non-adjacent. We claim that $G'$ has a $(k-1)$-clique if and ony if $\sigma$ is not a $k$-C-E ordering.
	
Suppose $G'$ has a $(k-1)$-clique $C$. The clique must be in a connected component of $G'$, say $G_{u,v}$. Then the vertices $u$ and $v$ are non-adjacent in $G$, and by construction of $G_{u,v}$, every vertex $w \in C$ is such that $w$ lies between $u$ and $v$ in $\sigma$ and $w$ is a neighbour to both $u$ and $v$. Thus $\{u,v\} \cup C$ forms an ordered $K_{k+1}^-$ in $\sigma$ and hence $\sigma$ is not a $k$-C-E ordering.
Conversely, suppose there is an ordered $K_{k+1}^-$ in $\sigma$. Let $Q$ be the vertices of the ordered $K_{k+1}^-$ and let $u$ and $v$ be its endpoints in $\sigma$. Then $C=Q\setminus \{u,v\}$ forms a $(k-1)$-clique such that every vertex $w \in C$ is a neighbour to both $u$ and $v$ and lies between $u$ and $v$ in $\sigma$. Therefore $C$ forms a $(k-1)$-clique in $G_{u,v}$. 
% If $G'$ has no $k-1$ clique, then no pair of non-adjacent vertices of $G$ is a conflicting pair, and hence
% $\sigma$ is a $k$-C-E ordering.
\end{proof}

 %It also gives us a $n^{O(k)}$ algorithm to verify whether the given ordering is a $k$-C-E ordering.
%
% \begin{corollary}\label{conp_complete}
% 	{\ProbVerify} is {\sc CoNP}-complete.
% \end{corollary}
%

%\begin{thm}
%Problem \ref{prob:verify} is W[1]-complete.
%\end{thm}
%\begin{proof}
%Follows from Theorem \ref{clique_to_verify} and Theorem \ref{verify_to_clique}.
%\end{proof}

% \subsection{Bounded Treewidth}
\bigskip

If all the $k$-cliques in a graph can be enumerated in time $O(f(k)poly(n))$ for some function $f$, then we can verify if an ordering is a $k$-C-E ordering in $O(f(k)poly(n))$ time by checking every pair of such cliques to see if they form an ordered $K_{k+1}^-$.
We show that a similar situation happens if $G$ has bounded treewidth
%(the number of $k$-cliques is a function of $k$ and the treewidth of the graph)
and so the verification problem becomes easy.

% \begin{lemma} \label{lem:clique_bag} (see for example \cite{pa_book})
% For any clique $K$ in $G$, there exists a vertex $v \in V(T)$ such that all the vertices of $K$ appear in the bag $B_v$ corresponding to the vertex $v$ in the tree decomposition.
% \end{lemma}
%
% Hence if $k-1$ is greater than the treewidth $tw$ of $G$, $G$ will be a trivial $(k-1)$-C-E graph as there can not be any $(k-1)$-sized cliques. So, for the problem to remain non-trivial, we assume that $k-1 \leq tw$.
%
% \begin{restatable}[$\star$]{theorem}{treewidthfpt}
% Given a graph $G$ on $n$ vertices, and a tree decomposition $\mathcal{T} = (T, \{X_t \}_{t \in V(T)})$ of $G$ of width $w$ that has $n^{O(1)}$ bags. We can verify whether an ordering is a $k$-clique-extendible ordering in $O(w^{k} poly(n))$ time.
% \end{restatable}
%
% Now, to obtain such a tree-decomposition of the graph, we can apply the following Theorem proven by Bodlaender et al.%Arnborg et. al and Robertson et. al. independently.
%
% \begin{lemma} \cite{treewidth_approx}
% There exists an algorithm, that given an $n$-vertex graph $G$ and an integer $k$, runs in time $2^{O(k)} \cdot n$ and either constructs a tree decomposition of $G$ of width at most $5k+4$ and $n^{O(1)}$ bags, or concludes that the treewidth of $G$ is greater than $k$.
% \end{lemma}
%
% Thus we have

\begin{lemma} \label{lem:clique_bag} (see for example \cite{pa_book})
For any clique $K$ in $G$, there exists a vertex $v \in V(T)$ such that all the vertices of $K$ appear in the bag $B_v$ corresponding to the vertex $v$ in the tree decomposition.
\end{lemma}

% \begin{restatable}[$\star$]{theorem}{treewidthfpt}
% Given a graph $G$ on $n$ vertices, and a tree decomposition $\mathcal{T} = (T, \{X_t \}_{t \in V(T)})$ of $G$ of width $w$ that has $n^{O(1)}$ bags. We can verify whether an ordering is a $k$-clique-extendible ordering in $O(w^{k} poly(n))$ time.
% \end{restatable}

\begin{lemma}[\cite{treewidth_approx}]\label{lem:find_tw}
There exists an algorithm, that given an $n$-vertex graph $G$ and an integer $t$, runs in time $2^{O(t)} \cdot n$ and either constructs a tree decomposition of $G$ of width at most $5t+4$ and $n^{O(1)}$ bags, or concludes that the treewidth of $G$ is greater than $t$.
\end{lemma}

\begin{theorem}
Given an ordering of the vertices of a graph $G$ on $n$ vertices, we can verify whether it is a $k$-C-E ordering of $G$ in time $O(tw^{O(tw)} poly(n))$, where $tw$ is the treewidth of $G$.
\end{theorem}
%\textbf{\textcolor{red}{TODO: add treewidth FPT for verify}}
\begin{proof}
Due to \cref{lem:clique_bag}, if $tw<k-1$, $G$ will be a trivial $k$-C-E graph as it cannot contain any $k$ sized cliques. So, for the problem to remain non-trivial, we assume that $k-1 \leq tw$.
	
We use \cref{lem:find_tw} to obtain a tree decomposition $\mathcal{T} = (T, \{X_t \}_{t \in V(T)})$ of $G$ of width $w\le 5tw+4$ that has $n^{O(1)}$ bags. We can verify whether an ordering is a $k$-clique-extendible ordering in $O(w^{k} poly(n))$ time as follows. From \cref{lem:clique_bag}, every $(k-1)$ sized clique must appear in one bag. As the bag sizes are bounded by $w+1$, and the number of bags is $n^{O(1)}$, we can enumerate all $(k-1)$-cliques within a bag in $w^k n^{O(1)}$ time. Now, for every such clique $K$ and for every pair of vertices $(u,v)$ in $G$ that are non-adjacent, we check whether $K \subseteq N(u) \cap N(v)$ and the set of vertices of $K$ appear between $u$ and $v$ in the ordering. If they do for at least one such clique $K$ and vertex pair $(u,v)$, we output ``no'', otherwise we output ``yes''.	

The algorithm takes $O(w^k poly(n))$ time. Since $k$ is upper bounded by $tw+1$ and $w = O(tw)$, this runtime is $O(tw^{O(tw)} poly(n))$.
\end{proof}

\section{Hardness of finding clique}\label{section:clique_in_kceo}

There exists an $n^{O(k)}$ algorithm for finding a maximum clique in a $k$-C-E graph~\cite{doublethreshold} when a $k$-C-E ordering is given. In this section, we will prove that this is most likely optimal, that is, we prove that unless ETH fails, there is no $f(k)n^{o(k)}$ algorithm for finding a maximum clique in a $k$-C-E graph \emph{even if the ordering is given.} 
We will reduce from the following problem.

\defproblem{Multicoloured Clique}{Graph $G$, a partition $V_1,\ldots,V_k$ of $V(G)$}{Does there exist a $k$-clique $C$ in $G$ such that $|C \cap V_i| = 1$ for each $i \in [k]$?}

\textsc{Multicoloured Clique} is W[1]-hard and cannot be solved in time $f(k)n^{o(k)}$ unless ETH fails~\cite{pa_book}. 
Given an instance $G,V_1,\ldots,V_k$ of {\sc Multicoloured Clique}, we will first remove all edges that lie within each partition $V_i$. Hence the graph is now $k$-colourable with colour classes $V_1,\ldots,V_k$. Any $k$-colourable graph is also a $k$-C-E graph by Observation \ref{obs:coloring}, and we can use an ordering $\phi$ of the form $V_1+V_2+\cdots+V_k$ to find the maximum clique size of $G$ using an algorithm to find maximum clique in a $k$-C-E graph. If the clique size  is equal to $k$, we output yes, otherwise output no. The following theorem follows.

\begin{theorem}\label{thm:clique_in_kceo_hard}
	Finding a maximum clique in a $k$-C-E graph, even if given a $k$-C-E ordering of the graph, is {\sc NP}-hard, {\sc W[1]}-hard and cannot be solved in time $f(k)n^{o(k)}$ unless {\sc ETH} fails.
\end{theorem}

\section{Finding a $k$-C-E Ordering} \label{section:ordering}

In this section, we consider the following problem and prove the main result of the paper.

% \begin{prob} \label{prob:ordering}
% Given a graph $G = (V,E)$ and an integer $k$, check whether $G$ has a $k$-clique-extendible ordering.
% \end{prob}

\defproblem{Find $k$-C-E Ordering}{Graph $G$, integer $k$}{Is $G$ a $k$-C-E graph?}

Note that this is possibly a harder problem than {\ProbVerify}, but still \cref{lem:verify_hard} doesn't immediately imply even {\sc co-W[1]}-hardness for this problem, as one may be able to determine whether $G$ has a $k$-C-E ordering without even verifying an ordering.
Our main result in this section is to show that {\ProbFind} is NP-hard for each $k \geq 3$. First we will show that {\ProbFind} is {\sc co-W[1]}-hard and {\sc co-NP}-hard. This result rules out algorithms running in time $f(k)n^{o(k)}$ assuming ETH (where as the NP-hardness rules out even $n^{f(k)}$ algorithms assuming P$\neq$NP).

% \subsection{{\sc Co-NP}-hardness and {\sc W[1]}-hardness}

%\unsure[inline]{This section proves W[1]-hardness of finding an ordering, but maybe we should remove this since we prove para np hardness anyway?}

\begin{theorem}\label{thm:ordering_w1hard}
{\ProbFind} is {\sc co-W[1]}-hard and {\sc co-NP}-hard.
\end{theorem}
\begin{proof}
We will reduce from the {\sc Clique} problem. Given an integer $k$ and a graph $G$ in which we wish to find a $k$ sized clique, we construct another graph $G'$ such that $G'$ contains the forbidden subgraph $F_k$ (defined in 
\cref{section:prelims}) if and only if it has a clique of size $k$. 

Let $V(G) = \{v_1,v_2,\ldots,v_n\}$. Let $K = \{v_{n+1},v_{n+2},\ldots,v_{n+k-1}\}$ be a $(k-1)$ sized clique such that each vertex in $K$ is connected to every vertex in $V(G)$ and let $I = \{u_{i,j} \mid i,j \in [n+k-1], i<j\}$ such that $u_{i,j}$ is adjacent to all vertices in $\{v_1,v_2,\ldots,v_{n+k-1}\}$ except for $v_i$ and $v_j$. Let $G'$ be the graph where $V(G') = V(G) \cup I \cup K$. % (See \cref{fig:find_whard}).
%We claim that $G$ has a $k$-clique if and only if $G'$ does not have a $k$-C-E ordering. 
\begin{claim}
$G$ has a $k$-clique if and only if $G'$ does not have a $k$-C-E ordering.
\end{claim}
\begin{proof}
See \cref{fig:find_whard} for a figure depicting the constructed graph $G'$. Suppose $S$ is a $k$-clique in $G$. Then $S \cup K \cup I$ will have an induced subgraph isomorphic to $F_k$ such that the $(2k-1)$-clique of $F_k$ is $K \cup S$ and the independent set of $F_k$ is a subset of $I$. By \cref{prop:forbidden}, $F_k$ does not have a $k$-C-E ordering and hence, by Observation \ref{obs:hereditary}, $G'$ does not have a $k$-C-E ordering.

Conversely if $G$ does not have a $k$-clique then any arbitrary ordering of $V(G)$ will be a $k$-C-E ordering of $G$. We will argue that any ordering $\pi$ of the form $V(G)+I +K$ is a $k$-C-E ordering of $G'$.

%\todo{Can you use some separator lemma to argue this? Then you will save this long paragraph and is a good use of the separator lemma.}
It is enough to prove that there does not exist an ordered $K_{k+1}^-$ in $\pi$. For contradiction, suppose $Q \subseteq V(G')$ forms an ordered $K_{k+1}^-$ in $\pi$. Let $a,b$ be the endpoints of $Q$ in $\pi$ so that $(a,b) \not\in E(G')$ and $a <_{\pi} b$. Let $A = Q\setminus \{b\}$ and $B = Q\setminus \{a\}$. Note that $A$ and $B$ are two $k$-cliques and $A \cup B=Q$. Since $I$ is an independent set (and $k \geq 2$), $Q$ contains at most one vertex from $I$. Therefore, if $a \in I \cup K$, then $B \subseteq K$, which is a contradiction as $|B|>|K|$. Similarly, if $b \in V(G) \cup I$, then $A \subseteq V(G)$, which is a contradiction as $A$ is then a $k$-clique in $G$. We thus have $a \in V(G)$ and $b \in K$. But then $(a,b) \in E(G')$, which is a contradiction.
%It cannot be the case that $A$ or $B$ is entirely contained in $V(G)$ since $G$ does not have any $k$-cliques. Furthermore, it can not be the case that $A$ or $B$ contains more than one vertex from $I$, since the vertices of $I$ are not adjacent to each other. Additionally, $K$ contains only $k-1$ vertices, so $A$ or $B$ can not be contained in $K$ either. So the vertices of $A$ and $B$ must be distributed between the three sets $V(G)$, $I$ and $K$, such that they contain not more than one vertex for $I$.
% Let $\{a\}  = A \setminus B$ and $\{b\}  = B \setminus A$. 
%Since $A \cup B$ is a set of size $k+1$, and at most $k$ vertices of $A \cup B$ can be in $I \cup K$, it must be the case that $a \in V(G)$. And, if $b \in V(G)$, then $V(G)$ contains the entirety of $B$ which is not possible, so $b \notin V(G)$. If $b \in I$, then since $A \cup B$ can have at most one vertex from $I$, $A$ will be contained in $V(G)$, which is again not possible, hence $b \notin I$. So, it must be the case that $b \in K$. Since $a \in V(G)$, and all vertices in $V(G)$ are adjacent to all vertices in $K$ by construction, this implies that $a$ is adjacent to $b$, a contradiction to our assumption.
Therefore, we conclude that there cannot be an ordered $K_{k+1}^-$ in $\pi$.
\end{proof}
\bigskip

The reduction maps the YES-instances of {\sc Clique} to the NO-instances of {\ProbFind} and vice-versa. Hence {\ProbFind} is {\sc co-W[1]}-hard and {\sc co-NP}-hard. 
\end{proof}

\begin{figure}[t!]
	\begin{center}
	\includegraphics[scale=0.7]{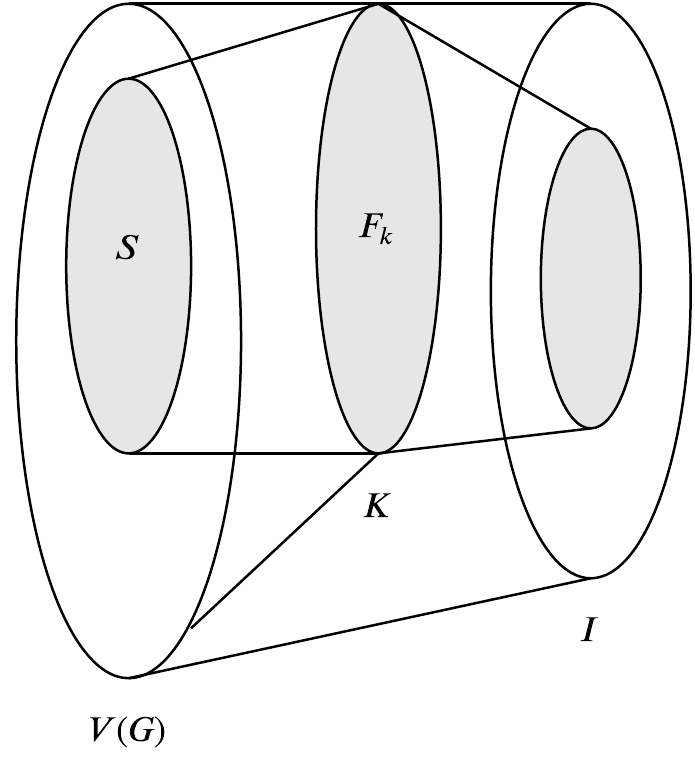}
	\end{center}
\caption{Diagram depicting the reduction for Theorem \ref{thm:ordering_w1hard}. The shaded region shows $F_k$ as an induced subgraph.}
\label{fig:find_whard}
\end{figure}

% We will now see that the problem is also NP-hard. Therefore if the problem is in NP or in {\sc co-NP}, it would imply that NP={\sc co-NP}, which is considered unlikely. Hence {\ProbFind} is unlikely to be in NP or {\sc co-NP}.

\subsection{NP-hardness for $k\ge 4$}

We now prove the NP-hardness of {\ProbFind} by a reduction from {\sc Betweenness} defined below. The reduction strategy works for all $k\ge 4$ but not for $k=3$ and so we give a different reduction for $k=3$ in the next section. %For $k\ge 4$, we reduce from the {\sc Betweeness} problem, defined as follows.

\defproblem{Betweenness}{Universe $U$ of size $n$, and a set of triples $\mathcal T = \{t_1,\ldots,t_m\}$ where each $t_i = (a_i,b_i,c_i)$ is an ordered triple of elements in $U$}{Does there exist an ordering $\phi$ of $U$ such that either $a_i <_{\phi } b_i <_{\phi } c_i$ or $c_i <_{\phi} b_i <_{\phi } a_i$ for each triple $(a_i,b_i,c_i) \in \mathcal T$?}

{\sc Betweenness} is NP-hard~\cite{betweeness_nphard}. To prove our reduction, we will require a gadget that takes as input a graph $G$ and 3 vertices $x,y,z \in V(G)$ and converts them to a modified graph $G'$ in such a way that either $x <_{\phi } y <_{\phi} z$ or $z <_{\phi } y <_{\phi } x$ for any $k$-C-E ordering $\phi $ of $G'$. Moreover, if $\phi $ is a $k$-C-E ordering of $G$ such that $x <_{\phi } y <_{\phi} z$ or $z <_{\phi } y <_{\phi } x$ then $\phi $ is also a $k$-C-E of $G'$. Thus the gadget `prunes' out the orderings of the graph where $y$ does not lie between $x$ and $z$ in the ordering. The $k$-C-E orderings of $G'$ are exactly the $k$-C-E orderings $\phi $ of $G$ where either $x <_{\phi } y <_{\phi} z$ or $z <_{\phi } y <_{\phi } x$.
Thus to construct the reduction, we will start with a graph where all $n!$ orderings are valid $k$-C-E orderings, and apply the gadget for each $(a_i,b_i,c_i) \in \mathcal T$. After applying the gadgets, we will have pruned out all the `bad' orderings and we will remain with exactly the set of orderings in which $b_i$ lies between $a_i$ and $c_i$ for each $i \in [m]$. 
To describe the construction of the gadget, first we need to define an auxiliary graph $\Gamma_k$. 

\bigskip
\noindent\textbf{Definition of the auxiliary graph.} Recall the graph $F_k$, defined in \cref{section:prelims} on the vertex set $K \cup I$ where $K=\{v_1,v_2,\ldots,v_{2k-1}\} $ induces a clique on $2k-1$ vertices, and every vertex in $I$ is indexed by a pair of vertices of $K$ to which the vertex is not adjacent.  Pick arbitrary vertices $v_1$ and $v_2$ of $K$ and let $u_{1,2}$ be the vertex of $I$ that is adjacent to every vertex of $K$ except $v_1$ and $v_2$.  Define $\Gamma_k = F_k \setminus \{u_{1,2}\}$.  
Note that $\Gamma_k$ has $O(k^2)$ many vertices.
% \unsure{do we still want replace $u_{1,2}$ with $u_{i,j}$?}
% That is, we remove the vertex in $I$, that is not adjacent to $v_1$ and $v_2$.
% \todo{The proofs of all the propositions here can go to appendix}
\begin{lemma}\label{prop:gamma_endpoints}
In any $k$-C-E ordering $\phi$ of $\Gamma_k$, $v_1$ and $v_2$ are the endpoints of $K$. Furthermore, there exists a $k$-C-E ordering $\phi$ of $\Gamma_k$ such that $v_1$ is the first element in $\phi$ and $v_2$ is the last.
\end{lemma}
\begin{proof}
Suppose that $\phi$ is a $k$-C-E ordering of $\Gamma_k$ and suppose for contradiction that $v_i$ and $v_j$ are the endpoints of $K$ where $\{i,j\} \neq \{1,2\}$. Then there exists $u_{i,j}\in I$ that is adjacent to every vertex in $K$ except $v_i$ and $v_j$. Let $\phi|_K = (v_{\sigma(1)}, v_{\sigma(2)},\ldots,v_{\sigma(2k-1)})$ where $\sigma$ is a permutation of $\{1,2,\ldots,2k-1\}$ such that ${\sigma(1)} = i$ and ${\sigma(2k-1)} = j$. 

If the vertex $u_{i,j}$ comes after $v_{\sigma(k)}$ in $\phi$, then the vertices in $\{v_{\sigma(1)}, v_{\sigma(2)},\ldots,$ $v_{\sigma(k)},u_{i,j}\}$ form an ordered $K_{k+1}^-$ in $\phi$. On the other hand, if the vertex $u_{i,j}$ comes before $v_{\sigma(k)}$ in $\phi$, then the vertices in $\{u_{i,j}, v_{\sigma(k)}, v_{\sigma(k+1)},\ldots,v_{\sigma(2k-2)},v_{\sigma(2k-1)}\}$ form an ordered $K_{k+1}^-$ in $\phi$.
In both cases, we get a contradiction to $\phi$ being a $k$-C-E ordering. Therefore every $k$-C-E ordering $\phi$ of $\Gamma_k$ is such that $v_1$ and $v_2$ are the endpoints of $K$.

Now we will prove the existence of an ordering $\phi$ of $\Gamma_k$ such that $v_1$ is the first element of $\phi$ and $v_2$ is the last. 
Let $I_1$ be the set of all vertices in $I$ that are not adjacent to $v_1$ and let $I_2=I\setminus I_1$. Note that, since we have removed the vertex $u_{1,2}$ from $F_k$ to get $\Gamma_k$, all vertices in $I_1$ are adjacent to $v_2$ and all vertices in $I_2$ are adjacent to $v_1$.

% Let $\psi$ be an ordering of $K$ such that $\psi(1) = v_1$ and $\psi(2k-1) = v_2$.

Consider an ordering $\phi$ of the form $(v_{1},v_{3},v_{4},v_5,\ldots, v_k) + I_1 +  (v_{k+1}) + I_2 + (v_{k+2},v_{k+3},\ldots, v_{2k-2},$ $v_{2k-1},v_2)$. We claim that $\phi$ is a $k$-C-E ordering of $\Gamma_k$. Observe that $v_1$ is the first element of $\phi$ and $v_2$ is the last, thus we will be done once we prove the claim.

Suppose that $\phi$ is not a $k$-C-E ordering of $\Gamma_k$, then there exists $Q \subseteq V(\Gamma_k)$ that induces an ordered $K_{k+1}^-$ in $\phi$. Let $a,b$ be the endpoints of $Q$ in $\phi$ so that $(a,b)\not\in E(\Gamma_k)$ and $a <_{\phi} b$. Let $A=Q\setminus \{b\} $ and $B=Q\setminus \{a\}$. Note that $A$ and $B$ are $k$-cliques and $A\cup B=Q$.%Note that, for $a_1$ and $a_2$ to be conflicting vertices, where $a_1 = \phi_i$ and $a_2 = \phi_j$, %then $|i - j| \geq k+1$, as there is a k-clique between $a_1$ and $a_2$.

%\unsure{Is this clear? there maybe a better way to say this?}
Note that at most one vertex from $I$ can be contained in $Q$ because otherwise $a,b\in I$, which implies that there is at most one vertex (which is $v_k$) between $a$ and $b$ in $Q$, contradicting the fact that $k\geq 3$.
If $a\in I$, then by the above observation, we have $B\subseteq\{v_{k+1},v_{k+2},\ldots,v_{2k-1},v_2\}$. Since $|B|=k$, we have $B=\{v_{k+1},v_{k+2},\ldots,v_{2k-1},v_2\}$, which implies that $a\in I_1$ and $b=v_2$. By our earlier observation, all vertices in $I_1$ are adjacent to $v_2$, which contradicts the fact that $ab\notin E(\Gamma_k)$. We thus have that $a\notin I$. By a symmetric argument, we get that $b\notin I$. Then $a,b\in K$, which again contradicts the fact that $ab\notin E(\Gamma_k)$.
%
%As $(a,b)$ is not an edge, both $a$ and $b$ cannot be in $K$ since $K$ is a clique, so either $a \in I$ or $b \in I$ (or both). Since $I$ is an independent set, and since every vertex in $A \cap B$ is adjacent to both $a$ and $b$, it must be that $(A \cap B) \cap I = \emptyset$. 
%
%If $b \in I_1$, since $a <_\phi b$, it must be that $a$ lies in $\{v_{1},v_3,v_4,\ldots,v_{k-1}\}$, and the elements of $A \cap B$ must lie in $\{v_3,v_4,\ldots,v_{k-1}\}$. There are only $k-2$ elements in $K$ that appear between $a$ and $b$ in $\phi$ but there are $k-1$ elements in $Q \setminus \{a,b\}$, all of which need to appear in $K$ and after $a$. So we have a contradiction.
%
%If $a \in I_1$ then $a$ is not adjacent to $v_1$ and $a$ is adjacent to $v_2$. If $b \neq v_2$, then there are less than $k-1$ elements in $K$ between $a$ and $b$, which leads to a contradiction using the same argument as above. So, it must be that $b = v_2$ but then $a$ is adjacent to $b$, which is a contradiction to our assumption. 
%
%If $a \in I_2 \cup I_3$ or $b \in I_2\cup I_3$, the proof is symmetric to the case when $b \in I_1$ or $a \in I_1$ respectively.
%
%In every case we get a contradiction so $\phi$ is a $k$-C-E ordering of $\Gamma_k$ such that $v_1$ is the first element and $v_2$ is the last.
\end{proof}
\bigskip

\noindent\textbf{The Gadget.} We will use $\Gamma_k$ as a gadget to constrict the set of orderings a graph can have. %We call this gadget the constrictor gadget. 
Pick an arbitrary vertex $v_3 \in K$ such that $v_3 \neq v_1,v_2$. Given a graph $G$, applying the gadget on a triplet of vertices $x,y,z \in V(G)$ involves taking the disjoint union of $G$ and $\Gamma_k$ and identifying the vertices $x$ with $v_1$, $y$ with $v_3$ and $z$ with $v_2$ (See Fig. \ref{fig:gamma}). For technical reasons, we will only be applying the gadget on vertices $x,y,z$ that induce a clique in $G$. 
% (Recall the notion of ``identification of a pair of vertices $x$ and $y$'' defined in \cref{section:prelims} which simply makes $x$ adjacent to all neighbors of $y$ and deletes $y$).
Since $\Gamma_k$ has $O(k^2)$ many vertices, the gadget will add $O(k^2)$ vertices to $G$, keeping it well within a polynomial factor.
% We can think of this gadget as an operation on a graph $G$ that takes in vertices $x,y,z$ as input and outputs a graph $G'$ with special properties. 
We use notation $G' = \mathcal C_k(G,x,y,z)$ to denote ``$G'$ is obtained by applying the gadget on $G$ on vertices $x,y,z$''. The valid $k$-C-E orderings of $G'$ should exactly be the $k$-C-E orderings of $G$ where $y$ does not come between $x$ and $z$. The following lemmas give us exactly that.

\begin{figure}[t]
\begin{center}
	\includegraphics[scale=0.5]{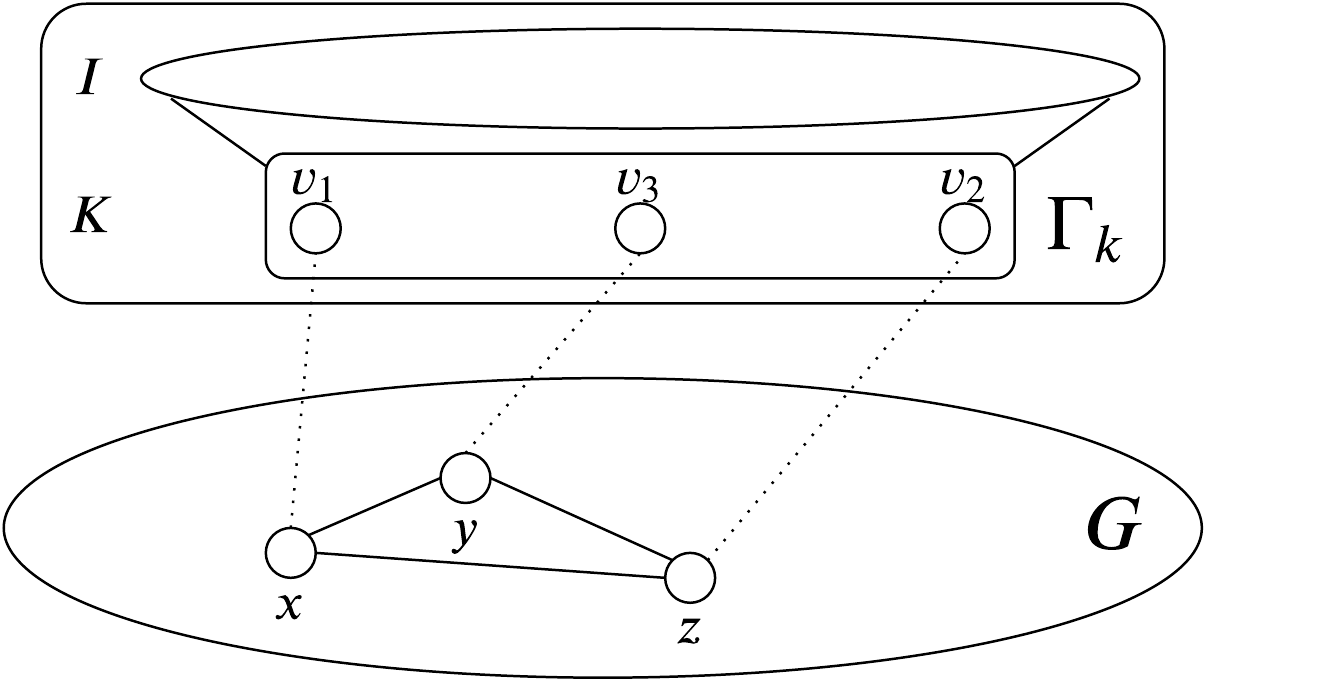}
	\caption{The construction of the gadget. Dotted lines indicate vertices identified to each other.} 
\label{fig:gamma}
\end{center}
\end{figure}

\begin{lemma}\label{prop:ct_constraint}
Let $G$ be a graph and let $x,y,z \in V(G)$ be vertices of $G$. In any $k$-C-E ordering $\phi$ of $G' = \mathcal C_k(G,x,y,z)$, $y$ comes between $x$ and $z$.
\end{lemma}
\begin{proof}
$\Gamma_k$ is a subgraph of $G'$. By Observation \ref{obs:hereditary}, any $k$-C-E ordering of $G'$ must induce a $k$-C-E ordering of $\Gamma_k$ within it. Furthermore, $v_3$ comes between $v_1$ and $v_2$ in any $k$-C-E ordering of $\Gamma_k$ by \cref{prop:gamma_endpoints} and since $x,y,z$ are identified with $v_1,v_3,v_2$ respectively, it follows that $y$ is between $x$ and $z$ in any $k$-C-E ordering of $G'$.
\end{proof}

\begin{lemma}\label{prop:gamma}
Let $k \geq 4$ and let $G$ be a graph that has a $k$-C-E ordering $\psi$ such that $y$ comes between $x$ and $z$ for some three vertices $x,y,z \in V(G)$ that form a $3$-clique in $G$, then $G' = \mathcal C_k(G,x,y,z)$ has a $k$-C-E ordering $\phi$ such that $\phi |_{V(G)} = \psi$.
\end{lemma}
\begin{proof}
By \cref{prop:gamma_endpoints}, $\Gamma_k$ has a $k$-C-E ordering $\sigma$ where $v_1$ and $v_2$ are the first and last elements in $\sigma$ respectively. Since $v_1$ is identified with $x$ and $v_2$ with $z$, we have that $x <_{\sigma} a <_{\sigma} z$ for each $a \in V(\Gamma _k)\setminus\{x,z\}$. We wish to use \cref{lem:separator2} on $\psi$ and $\sigma$ to obtain a $k$-C-E ordering of $G'$. Observe that $V(G) \cup V(\Gamma_k) = V(G')$ and $V(G) \cap V(\Gamma_k) = \{x,y,z\} $ separates $V(G)$ and $V(\Gamma _k)$, thus the first and second conditions in \cref{lem:separator2} hold. Since $x,y,z$ form a clique in both $G$ and in $\Gamma _k$, we have that $G'[V(G)]$ is isomorphic to $G$ and $G'[V(\Gamma _k)]$ is isomorphic to $\Gamma _k$. Thus $\psi$ is a $k$-C-E ordering of $G'[V(G)]$ and $\sigma$ is a $k$-C-E ordering of $G'[V(\Gamma _k)]$. Without loss of generality, we can assume by Observation \ref{obs:reverse}, that $x <_{\psi}y <_\psi z$ in $G$. It also holds that, $x<_{\sigma} y <_{\sigma} z$ since $y$ is identified with $v_3 \in V(\Gamma _k)\setminus \{x,z\} $. Therefore $\psi|_{\{x,y,z\}} = \sigma|_{\{x,y,z\}}$, and the third condition also holds. Let $a$ be a vertex in $V(\Gamma _k)\setminus \{x,y,z\}$. Suppose there exists a clique $C \subseteq V(G)\cap V(\Gamma _k) = \{x,y,z\} $ of size $k-1$ that is adjacent to $a$. Since $k \ge 4$, it follows that $|C| \ge 3$, and thus $x,z \in C$ are the endpoints of $C$ in $\sigma$. By the property of $\sigma$, we have $x<_{\sigma} a <_{\sigma} z$. Thus all four conditions for \cref{lem:separator2} are satisfied and the lemma follows.
\end{proof}

%\subsection{NP-hardness}
\bigskip

\noindent\textbf{The Reduction.} We are now ready to prove that the problem of checking whether a graph has a $k$-C-E ordering is NP-hard for each $k \geq 4$.

\begin{theorem} \label{thm:paranphardness}
{\ProbFind} is {\sc NP}-hard for each $k \geq 4$.
\end{theorem}
\begin{proof}
We will reduce from {\sc Betweenness}. Let $\mathcal I = (U,\mathcal T)$ be the input {\sc Betweenness} instance. %where $U$ is the set of elements and $\mathcal T = \{t_1,\ldots,t_m\}$ where $t_i = (a_i,b_i,c_i)$ is the $i$th triple.  
We want to construct a graph $G'$ such that $G'$ has a $k$-C-E ordering if and only if the {\sc Betweenness} instance is satisfiable. 
% We need to either output an ordering $\phi $ such that for each $t_i \in \mathcal T$, $b_i$ comes between $a_i$ and $c_i$ in $\phi $, or conclude that no such ordering exists. 
We will construct a graph with vertex set equal to the universe $U$ and apply the gadget for every triple $(a_i,b_i,c_i)$ in $\mathcal T$. We do this iteratively, that is, we first define $G_0$ to be the complete graph on vertex set $U$ and then construct $G_i = \mathcal C_k(G_{i-1},a_i,b_i,c_i)$ for each $i \in [m]$ (where $m$ is the number of triples in $\mathcal T$). The final graph $G'$ is equal to $G_m$. 
There are $m$ many calls to the gadget and each gadget adds $O(k^2)$ vertices to $G'$. So the final size of $G'$ is $O(n+m k^2)$ (where $n$ is the size of $U$), which is polynomial in $n$ and $m$. 

\begin{claim}\label{claim:paranphard}
$G'$ has a $k$-C-E ordering if and only if $\mathcal I$ is a {\sc Yes}-instance.
\end{claim}
\begin{proof}
$\bm{(\Rightarrow)}$ Recall that the graph $G'$ is constructed iteratively as follows. We first define $G_0$ to be the complete graph on the vertex set $U$ and then $G_i = \mathcal C_k(G_{i-1},a_i,b_i,c_i)$ for $i \in [m]$. The final graph $G'$ is $G_m$. Suppose $G'$ has a $k$-C-E ordering $\phi$. Let $\psi = \phi|_U$. We claim that $\psi$ is a valid betweenness ordering for the instance $\mathcal I$. Let $(a_i,b_i,c_i) \in \mathcal T$ be a triple. Consider the subgraph $G_i$ of $G'$. By \cref{prop:ct_constraint}, $b_i$ comes between $a_i$ and $c_i$ in any $k$-C-E ordering of $G_i$. Since $G_i$ is a subgraph of $G$, by Observation \ref{obs:hereditary}, $b_i$ must come between $a_i$ and $c_i$ in any $k$-C-E ordering of $G'$ as well. This holds for every triple $(a_i,b_i,c_i) \in \mathcal T$, thus $\psi$ is a valid betweenness ordering for the instance $\mathcal I$. 

$\bm{(\Leftarrow)}$  Let $\psi$ be a valid betweenness ordering of $U$. We will prove that there exists a $k$-C-E ordering $\phi$ of $G'$ such that $\phi|_U=\psi$. Clearly, $G_0$ has such an ordering. We proceed by induction on $i=1,\ldots,m$. Suppose, by the induction hypothesis, that $G_i$ has a $k$-C-E ordering $\phi_i$ such that $\phi_i|_U=\psi$. Note that $b_i$ comes between $a_i$ and $c_i$ in $\psi$ (and thus also in $\phi_i$) as it is a valid betweenness ordering. 
% The induction step then follows by \cref{prop:gamma}. 
By \cref{prop:gamma}, $G_{i+1}$ has a $k$-C-E ordering $\phi_{i+1}$ such that $\phi_{i+1}|_{V(G_i)} = \phi_i$ and thus also $\phi_{i+1}|_U = \phi_i|_U = \psi$. Thus the induction step follows.
\end{proof}
\bigskip

The theorem follows from the above claim.
\end{proof}
\bigskip

The above reduction shows that {\ProbFind} is NP-hard. From \cref{thm:ordering_w1hard}, {\ProbFind} is also {\sc co-NP}-hard. Thus it is unlikely that the problem is in NP or in {\sc co-NP}. Moreover, it is easy to verify that the problem lies in $\Sigma_2$, as one can simply guess the ordering $\phi $ and use a {\sc co-NP} machine (\cref{lem:verify_hard}) to check whether $\phi $ is a $k$-C-E ordering. Thus it is an open question whether {\ProbFind} is $\Sigma _2$-complete.
\bigskip

\noindent\textbf{Remark (1)} It is important to note that the problem is \emph{not} {\sc co-NP}-hard when $k$ is a fixed constant as opposed to it being given as a input. When $k$ is fixed, the $k$-C-E ordering itself is an NP certificate for the problem, as given an ordering it is easy to check whether it is a $k$-C-E ordering for constant $k$.
Thus, when $k$ is constant, the problem is NP-complete.
Indeed, the proof of {\sc co-NP}-hardness in \cref{lem:verify_hard} assumes that $k$ is given as an input. 
\medskip

\noindent\textbf{Remark (2)} The reduction in \cref{thm:paranphardness} does not work for $k=3$ due to technicalities that arise in order to satisfy the fourth condition of \cref{lem:separator2}, due to which we require that $|V(G) \cap V(\Gamma _k)| \le k-1$ (see proof of \cref{prop:gamma}). Since $|V(G) \cap V(\Gamma _k)| = |\{x,y,z\}|  = 3$ in the gadgets we construct, this forces $k$ to be at least $4$. We give a separate proof for NP-hardness of $k=3$ in the following section that uses some different ideas.

\subsection{NP-hardness for $k=3$}

In this section, we prove that the problem of finding a $3$-C-E ordering is NP-hard. We will reduce from the {\sc $3$-Colouring} problem.
Given a graph $G$ and an ordering $\phi$ of $V(G)$, we say that three edges $(u,v),(w,x),(y,z)\in E(G)$ form a \emph{disjoint triple} in $\phi$ if $u<_{\phi}v\leq_{\phi} w<_{\phi} x\leq_{\phi} y<_{\phi}z$. Here $x \leq_{\phi} y$ means that either $x=y$ or $x <_{\phi} y$. 
% The graph $K^-_4$ is known as a ``diamond'' in the literature; for the purposes of this section, we shall refer to an ordered $K^-_4$ as an ``ordered diamond''. Thus, an ordering $\phi$ of $V(G)$ is a 3-C-E ordering of $G$ if and only if it does not contain an ordered diamond.
\begin{obs}\label{obs:3-coldisj}
Let $G$ be a $3$-colourable graph and let $C_1,C_2,C_3$ be a partition of $V(G)$ into three independent sets. Then any ordering $\phi$ of $V(G)$ of the form $C_1+C_2+C_3$ contains no disjoint triple.
\end{obs}
\begin{proof}
Suppose that $(u,v),(w,x),(y,z)\in E(G)$ form a disjoint triple in $\phi$, where $u<_{\phi}v\leq_{\phi} w<_{\phi}x\leq_{\phi} y<_{\phi} z$. For a vertex $a\in V(G)$, let $c(a)$ denote the integer $i\in\{1,2,3\}$ such that $a\in C_i$. Since $u<_{\phi}v$, $(u,v)\in E(G)$, and $\phi$ is of the form $C_1+C_2+C_3$, it must be the case that $c(v)\geq c(u)+1$. Since $v\leq_{\phi} w$, we then get $c(w)\geq c(v)\geq c(u)+1$. Similarly, as $(w,x)\in E(G)$, $c(x)\geq c(w)+1\geq c(u)+2$, and further, as $x\leq y$, we get $c(y)\geq c(u)+2$. Continuing in this fashion, since $(y,z)\in E(G)$, we get $c(z)\geq c(u)+3$. But as $c(u)\geq 1$, we now have $c(z)>3$, which is a contradiction.
\end{proof}

\begin{obs}\label{obs:disj3-col}
Let $G$ be any graph. If there is an ordering $\phi$ of $V(G)$ that contains no disjoint triple, then $G$ is $3$-colourable.
\end{obs}
\begin{proof}
We orient the edges of $G$ using $\phi$. The edge $(u,v)$ is oriented from $u$ to $v$ if $u <_{\phi} v$. Let $S$ be the set of vertices which have no incoming arc after the orientation. For a vertex $u$, let $d(u)$ be the maximum distance from a vertex in $S$ to $u$ in the oriented graph. Since there is no disjoint triple in $\phi$, there cannot be a directed path of length $3$ in the oriented graph, thus $d(u) \le 2$ for each $u$. Let $C_i = \{u  \mid d(u)=i\}$. We prove that $C_i$ forms an independent set for each $i \in \{0,1,2\}$ thus $G$ is $3$-colourable. Suppose, for contradiction, there is an edge $(v,w) \in G[C_i]$. Without loss of generality assume that $v <_{\phi} w$ thus the edge is oriented from $v$ to $w$. Now there is a path from $S$ to $v$ of length $i$ and then taking the arc $(v,w)$ forms a path of length $i+1$ to $w$, contradicting that $w \in C_i$.  
\end{proof}
\bigskip

% \begin{proof}
% 	Let $a=\min_\phi\{v \mid \exists u<_{\phi} v$ such that $(u,v)\in E(G)\}$ and $b=\min_\phi\{v \mid \exists u$ such that $a\leq_{\phi} u<_{\phi} v$ and $(u,v)\in E(G)\}$. We claim that $A=\{u  \mid  u<_{\phi} a\}$, $B=\{u \mid  a\leq_{\phi} u<_{\phi} b\}$, and $C=\{u \mid  b\leq_{\phi} u\}$ are three independent sets in $G$ that form a partition of $V(G)$. Clearly $\{A,B,C\}$ form a partition of $V(G)$, so it only needs to be proven that each is an independent set in $G$. There cannot be two vertices $x,y\in A$ such that $(x,y)\in E(G)$, because if that were the case then $a\leq_{\phi}\max\{x,y\}$, implying that $\max\{x,y\}\notin A$. Similarly, there cannot be two vertices $x,y\in B$ such that $(x,y)\in E(G)$, as if that were the case then $b\leq_{\phi}\max\{x,y\}$, implying that $\max\{x,y\}\notin B$. Now suppose that $C$ is not an independent set. Then there exist $x,y\in C$ such that $(x,y)\in E(G)$. By the definition of $a$ and $b$, there exist $a',b'\in V(G)$ such that $(a',a)\in E(G)$, $a'<_{\phi} a$, $(b',b)\in E(G)$, and $a\leq_{\phi} b'<_{\phi} b$. Then the edges $(a',a)$, $(b',b)$ and $(x,y)$ form a disjoint triple in $\phi$, which is a contradiction.
% \end{proof}
% \medskip

It follows from Observations~\ref{obs:3-coldisj} and~\ref{obs:disj3-col} that a graph $G$ is 3-colourable if and only if there is an ordering of its vertex set containing no disjoint triple.
% \medskip

Another observation is that in any $3$-C-E ordering $\phi$ of $G$, for any pair of non-adjacent vertices $u,v \in V(G)$, the vertices that are adjacent to both $u$ and $v$ and lie between $u$ and $v$ in $\phi$ must be an independent set in $G$. Indeed, if there is an edge $(a,b)$ such that $u <_{\phi} a <_{\phi} b <_{\phi} v$ and $a,b$ are adjacent to both $u$ and $v$, then $G[\{u,a,b,v\}]$ is an ordered $K_4^-$ in $\phi$. This suggests a reduction from {\sc 3-Colouring}. The idea is that, associated to every edge $e=(u,v) \in E(G)$, we will add a vertex $t^e_1$, and a pair of adjacent vertices $t^e_2$ and $t^e_3$.
We will add edges so that the $t_2$ vertices and $t_3$ vertices together form a clique and the $t_1$ vertices form an independent set. We also add edges between all $t_2,t_3$ vertices and $t_1$ vertices. %Using gadgets we will ensure that, for every edge $(u,v)$, the vertices $t_1^{(u,v)},t_2^{(u,v)},t_3^{(u,v)}$ lie between $u$ and $v$ in every $k$-C-E ordering.
We will add a gadget to ensure that $t_1^e,t^e_2,t^e_3$ all lie between $u$ and $v$ in any 3-C-E ordering of $G'$. 

If $G$ is not 3-colourable, then for any ordering $\phi$ of $V(G')$, there will be a disjoint triple in $\phi|_{V(G)}$. If the disjoint triple is formed by the edges $(u,v),(w,x),(y,z)$ of $G$, where $u<_{\phi} v\leq_{\phi} w<_{\phi} x\leq_{\phi} y<_{\phi} z$, then the vertices $t^{(u,v)}_1,t^{(w,x)}_2,t^{(w,x)}_3,t^{(y,z)}_1$ form an ordered $K_4^-$ in $\phi$, and hence there can be no 3-C-E ordering of $G'$. On the other hand, our construction makes sure that if $G$ is a 3-colourable graph, then there exists a 3-C-E ordering for $G'$.
% We then prove that if a disjoint triple exists in $G$ then an ordered diamond will exist in $G'$. 
We now describe the reduction in detail.
% \medskip

% \begin{figure}[t!]
% 	\centering
% 	\includegraphics[width=0.8\textwidth]{3-ceo-figure.pdf}
% 	\caption{The Construction}
% 	\label{fig:3_ceo_fig}
% \end{figure}
\medskip

\noindent\textbf{The Construction.} Given a graph $G$, we construct a supergraph $G'$ as explained below (also see Figure~\ref{fig:construction}). For subsets $A,B \subseteq V(G)$, by ``join $A$ and $B$'', we mean that we add all possible edges between vertices in $A$ and vertices in $B$.  
To construct the vertex set of $G'$, we take the vertex set of $G$ and add the following.
\begin{enumerate}
	\item Add 4 sets of vertices $A = \{a,a_1,a_2,a_3\}$, $B=\{b,b_1,b_2,b_3\}$, $C=\{c,c_1,c_2,c_3\}$ and  $D=\{d,d_1,d_2,d_3\}$
	\item Add the sets of vertices $F = \{f^e_i  \mid e \in E(G), i \in \{1,2,\ldots,6\}\}$ and $T = \{t^e_i  \mid e \in E(G), i \in \{1,2,3\}\}$ 
\end{enumerate}

To construct the edge set of $G'$, we take the edge set of $G$ and add the following.

\begin{enumerate}
\item Add edges to make $A,B,C$ and $D$ into cliques of size $4$ each.
\item Add edges to make $t^e_i,f^e_{2i-1},f^e_{2i}$ into a clique, for each edge $e \in E(G)$ and $i \in [3]$
\item Join $\{a_1,a_2,a_3\}$ and $\{b_1,b_2,b_3\}$
\item Join $\{b_1,b_2,b_3\}$ and $\{c_1,c_2,c_3\}$
\item Join $\{c_1,c_2,c_3\}$ and $\{d_1,d_2,d_3\}$
\item Join $\{d_1,d_2,d_3\}$ and $\{a_1,a_2,a_3\}$
\item Join $\{a_1,a_2,a_3,b_1,b_2,b_3\}$ and $V(G)$
\item Join $\{c_1,c_2,c_3,d_1,d_2,d_3\}$ and $V(G) \cup F$
\item Add edges $(f^{(u,v)}_i,u)$ and $(f^{(u,v)}_i,v)$, for each $(u,v) \in E(G)$ and $i \in \{1,\ldots,6\} $
\item Add edges to make $\bigcup_{e \in E(G)} \{t_2^e,t_3^e\}$ into a clique
\item Join $\bigcup_{e \in E(G)} t_1^e$ and $\bigcup_{e \in E(G)} \{t_2^e,t_3^e\}$
\end{enumerate}

\begin{figure}[t]
\begin{center}
	\includegraphics[scale=0.85]{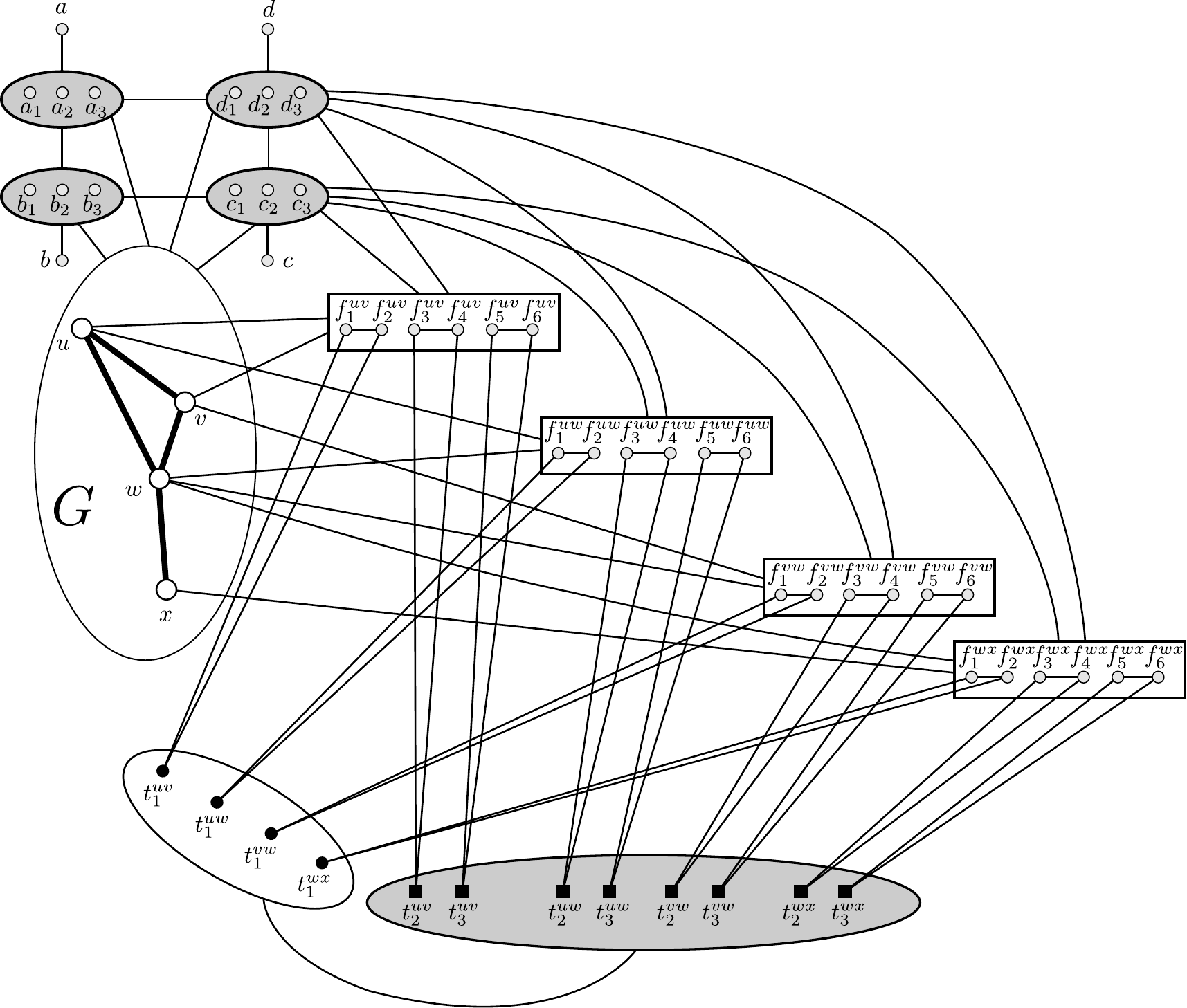}
\caption{The construction of $G'$ from $G$. The vertices inside each shaded block form a clique. An edge between a vertex $u$ and a block means that $u$ is adjacent to every vertex in the block, and an edge between two blocks means that every vertex in one block is adjacent to every vertex in the other block. Note that an edge between vertices $a$ and $b$ is denoted as $ab$ instead of $(a,b)$ to reduce clutter.
}\label{fig:construction}
\end{center}
\end{figure}

% Let $A = \{a_1,a_2,a_3,a\}$, $B=\{b_1,b_2,b_3,b\}$, $C=\{c_1,c_2,c_3,c\}$ and $D=\{d_1,d_2,d_3,d\}$. Let $F = \{f^e_i  \mid e \in E(G), i \in \{1,2,\ldots,6\}\}$ and $T = \{t^e_i  \mid e \in E(G), i \in \{1,2,3\}\}$. 
% Refer to Figure~\ref{fig:construction} for a schematic representation of the construction of $G'$ from $G$.

\begin{lemma}\label{lem:G'toG}
If $G'$ has a $3$-C-E ordering then $G$ is $3$-colourable.
\end{lemma}
\begin{proof}
Suppose that $\phi$ is a 3-C-E ordering of $G'$. By Observation~\ref{obs:disj3-col}, we only need to show that there is no disjoint triple in the ordering $\phi|_{V(G)}$. We can assume without loss of generality that there exist distinct $i,j\in\{1,2,3\}$ such that in the ordering $\phi$, we have $a_i<_{\phi}a_j<_{\phi}a$ (reversing the ordering $\phi$ if necessary; recall Observation~\ref{obs:reverse}). If there is a vertex $w\in V(G)\cup\{b_1,b_2,b_3,d_1,d_2,d_3\}$ such that $w<_{\phi} a_i$, then $w,a_i,a_j,a$ form an ordered $K_4^-$ in $\phi$, which contradicts the fact that $\phi$ is a 3-C-E ordering. 
Therefore, we can assume without loss of generality that the vertex $a_1$ occurs before every vertex of $V(G)\cup\{b_1,b_2,b_3,d_1,d_2,d_3\}$ in the ordering $\phi$. This also means that if there exist distinct $i,j\in\{1,2,3\}$ such that $b_i<_{\phi}b_j<_{\phi}b$, then $a_1,b_i,b_j,b$ would form an ordered $K_4^-$ in $\phi$. Thus, we conclude that there exist distinct $i,j\in\{1,2,3\}$ such that $b<_{\phi} b_i<_{\phi}b_j$, and arguing as before, we assume without loss of generality that the vertex $b_1$ occurs after every vertex of $V(G)\cup\{a_1,a_2,a_3,c_1,c_2,c_3\}$ in the ordering $\phi$. Now if there exist distinct $i,j\in\{1,2,3\}$ such that $c<_{\phi} c_i<_{\phi} c_j$, then $c,c_i,c_j,b_1$ form an ordered $K_4^-$ in $\phi$. Thus, there exist distinct $i,j\in\{1,2,3\}$ such that $c_i<_{\phi} c_j<_{\phi} c$, and reasoning as before, we can assume without loss of generality that $c_1$ occurs before every vertex in $V(G)\cup\{b_1,b_2,b_3,d_1,d_2,d_3\}$. Using similar arguments, we conclude that $d_1$ occurs after every vertex in $V(G)\cup\{a_1,a_2,a_3,c_1,c_2,c_3\}$ in $\phi$.

\begin{claim}\label{clm:finedge}
For every edge $(u,v)\in E(G)$, all the vertices in $\{f^{(u,v)}_i\colon 1\leq i\leq 6\}$ occur between $u$ and $v$ in $\phi$.
\end{claim}
\begin{proof}
Suppose that there exists $i\in\{1,2,\ldots,6\}$ such that $f^{(u,v)}_i<_{\phi} u<_{\phi} v$. Then the vertices $f^{(u,v)}_i,u,v,b_1$ form an ordered $K_4^-$ in $\phi$, which contradicts the fact that $\phi$ is a 3-C-E ordering. Similarly, if $u<_{\phi}v<_{\phi}f^{(u,v)}_i$ for some $i\in\{1,2,\ldots,6\}$, then $a_1,u,v,f^{(u,v)}_i$ form an ordered $K_4^-$ in $\phi$; again a contradiction. %This proves the claim.
\end{proof}

\begin{claim}\label{clm:ginf}
For every edge $e\in E(G)$ and $i\in\{1,2,3\}$, the vertex $t^e_i$ occurs between $f^e_{2i-1}$ and $f^e_{2i}$ in $\phi$.
\end{claim}
\begin{proof}
Since in the ordering $\phi$, $c_1$ occurs before every vertex in $V(G)$ and $d_1$ occurs after every vertex in $V(G)$, it follows from the above claim that $c_1$ occurs before every vertex in $F$ and $d_1$ occurs after every vertex in $F$. Now suppose that for some $e\in E(G)$ and $i\in\{1,2,3\}$, we have $f^e_{2i-1},f^e_{2i}<_{\phi} t^e_i$. Then the vertices $c_1,f^e_{2i-1},f^e_{2i},t^e_i$ form an ordered $K_4^-$ in $\phi$, which is a contradiction. Similarly, if $t^e_i<_{\phi} f^e_{2i-1},f^e_{2i}$, then the vertices $t^e_i,f^e_{2i-1},f^e_{2i},d_1$ form an ordered $K_4^-$ in $\phi$, again a contradiction.% This proves the claim.
\end{proof}
\bigskip

From the above two claims, it follows that for any edge $(a,b)\in E(G)$ and $i\in\{1,2,3\}$, the vertex $t^{(a,b)}_i$ occurs between $a$ and $b$ in $\phi$. Now suppose for the sake of contradiction that $(u,v),(w,x),(y,z)\in E(G)$ form a disjoint triple in $\phi|_{V(G)}$, where $u<_{\phi} v\leq_{\phi} w<_{\phi}x\leq_{\phi} y<_{\phi}z$. Then we have $u<_{\phi}t^{(u,v)}_1<_{\phi}v\leq_{\phi} w<_{\phi}t^{(w,x)}_2,t^{(w,x)}_3<_{\phi}x\leq_{\phi} y<_{\phi}t^{(y,z)}_1<_{\phi}z$. But then the vertices $t^{(u,v)}_1,t^{(w,x)}_2,t^{(w,x)}_3,t^{(y,z)}_1$ form an ordered $K_4^-$ in $\phi$, a contradiction. %This completes the proof.
\end{proof}
\medskip

\begin{lemma}\label{lem:GtoG'}
If $G$ is $3$-colourable then $G'$ has a $3$-C-E ordering.
\end{lemma}
\begin{proof}
Let $M = V(G) \cup F$, $L = M \cup A \cup B \cup C \cup D$ and $U = F \cup T$. Note that $L \cap U = F$ and in fact $F$ separates $L$ and $U$. The idea is to apply \cref{lem:separator2} on the subgraphs $L$ and $U$. 
Let $V_1,V_2,V_3$ be a partition of $V(G)$ into three colour classes. Let $E_1 = \{(u,v) \in E(G)  \mid u \in V_1\}$ and $E_2 = \{(u,v) \in E(G)  \mid u \in V_2, v \in V_3\}$. Observe that $\{E_1,E_2\}$ is a partition of $E(G)$. Let $\sigma$ be an ordering of $M$ of the form

\[V_1 + \sum_{i \in 3,2,1} \sum_{e \in E_1} (f^e_{2i},f^e_{2i-1}) + V_2+ \sum_{i \in 1,2,3} \sum_{e \in E_2} (f^e_{2i-1},f^e_{2i}) + V_3\]
Here the sum notation denotes iterated concatenation. For example, the term $\sum_{i \in 3,2,1} \sum_{e \in E_1} (f^e_{2i},f^e_{2i-1})$ would be a shorthand for 

\[(f^{e_1}_6,f^{e_1}_5,f^{e_2}_6,f^{e_2}_5,\ldots,f^{e_1}_4,f^{e_1}_3,f^{e_2}_4,f^{e_2}_3,\ldots,f^{e_1}_2,f^{e_1}_1,f^{e_2}_2,f^{e_2}_1,\ldots)\]
where $\{e_1,e_2,\ldots\}$ are the edges in $E_1$. Similarly, the term $\sum_{i \in 1,2,3} \sum_{e \in E_2} (f^e_{2i-1},f^e_{2i})$ would be

\[(f^{e^1}_1,f^{e^1}_2,f^{e^2}_1,f^{e^2}_2,\ldots,f^{e^1}_3,f^{e^1}_4,f^{e^2}_3,f^{e^2}_4,\ldots,f^{e^1}_5,f^{e^1}_6,f^{e^2}_5,f^{e^2}_6,\ldots)\]
where $\{e^1,e^2,\ldots\}$ are the edges in $E_2$.
Let $\phi$ be the following ordering of $L$.

\[\phi =(a_1,a_2,a_3,a,c_1,c_2,c_3,c) + \sigma + (d,d_1,d_2,d_3,b,b_1,b_2,b_3)\]
We also define an ordering $\psi$ of $U$ as follows.

\[\psi = \sum_{i \in 3,2,1} \sum_{e \in E_1} (f^e_{2i},t^e_{i},f^e_{2i-1})+ \sum_{i \in 1,2,3} \sum_{e \in E_2} (f^e_{2i-1},t^e_i,f^e_{2i})\]

Note that $\psi|_F = \sigma|_F = \phi|_F$. We claim that $\phi$ is a 3-C-E ordering of $G'[L]$ and $\psi$ is a 3-C-E ordering of $G'[U]$. If the claim is true, we are done as the lemma statement will follow from an application of \cref{lem:separator2}. Note that $\psi$ has the property that, if $a \in T$ is adjacent to a $2$-clique $(u,v)$ in $F$, then $u <_{\psi} a <_{\psi} v$. Thus $\psi$ satisfies the last condition for \cref{lem:separator2}.

\begin{claim}
$\psi$ is a 3-C-E ordering of $G'[U]=G'[T \cup F]$.
\end{claim}
\begin{proof}
Suppose that there exists an ordered $K_4^-$ having vertices $u,x,y,v$, such that $u <_{\psi} x <_{\psi} y <_{\psi} v$. Since every vertex in $F$ has degree 2 in $G'[T\cup F]$, only the endpoints of $\{u,x,y,v\}$ can belong to $F$. Thus $\{x,y\}\subseteq T$. Now, the fact that every vertex in $F$ is adjacent to exactly one vertex in $T$ implies that $\{u,x,y,v\} \subseteq T$. Since $(u,v) \not\in E(G')$, we have that $u,v \in \bigcup_{e \in E(G)} \{t_1^e\}$, which implies that $x,y \in \bigcup_{e \in E(G)} \{t_2^e,t_3^e\}$. But in $\psi$, it is impossible for two vertices from $\bigcup_{e \in E(G)} \{t_2^e,t_3^e\}$ to occur between two vertices in $\bigcup_{e \in E(G)} \{t_1^e\}$. Thus we have a contradiction.        
\end{proof}

\begin{claim}
$\sigma$ is a 3-C-E ordering of $G'[M]=G'[V(G) \cup F]$.
\end{claim}
\begin{proof}
Suppose that there exists an ordered $K_4^-$ having vertices $u,x,y,v$ such that $u <_{\sigma} x <_{\sigma} y <_{\sigma} v$. Note that $\sigma|_{V(G)}$ is of the form $V_1+V_2+V_3$, and is therefore a 3-C-E ordering of $V(G)$ by Observation~\ref{obs:coloring}. Thus $\{u,x,y,v\}$ is not contained in $V(G)$, implying that $\{u,x,y,v\}\cap F\neq\emptyset$.
Since $G'[F]$ is a collection of disjoint edges, it follows that $|\{u,x,y,v\} \cap F| \le 2$. If $F$ intersects $\{u,x,y,v\}$ at a single vertex, since every vertex in $F$ has at most $2$ neighbours in $V(G)$, that vertex must be an endpoint of $\{u,x,y,v\}$ and its neighbours in $V(G)$ must be $x,y$, a contradiction to the fact that every vertex in $F$ lies between its two neighbours in $V(G)$ in $\sigma$.
Thus $|\{u,x,y,v\} \cap F| = 2$. If the two vertices in $\{u,x,y,v\}\cap V(G)$ occur consecutively in $u,x,y,v$, then they have common neighbour in $F$ that does not lie between them in $\sigma$, which is a contradiction. If in $u,x,y,v$, there is exactly one vertex in $F$ between the two vertices in $V(G)$, then it means that in $\sigma$, there is a vertex in $V(G)$ between two adjacent vertices in $F$, which is again a contradiction.
Thus we have that $x,y \in F$ and $u,v\in V(G)$. But now $u$ and $v$ are two neighbours of $x\in F$ in $V(G)$ that are not adjacent to each other, which is again a contradiction.
\end{proof}

\begin{claim}
$\phi$ is a 3-C-E ordering of $G'[L]=G'[M\cup A\cup B\cup C\cup D]$.
\end{claim}
\begin{proof}
Suppose that there exists an ordered $K_4^-$ having vertices $u,x,y,v$ such that $u <_{\phi} x <_{\phi} y <_{\phi} v$. As $a$ has no neighbours to its right and no non-neighbours to its left in $\phi$, we have $a\notin\{u,x,y,v\}$. Symmetrically, $b\notin\{u,x,y,v\}$. If $c\in\{u,x,y,v\}$, then since in $\phi$, $c$ has no neighbours to its right and its only neighbours to the left are $c_1,c_2,c_3$, we have $v=c$, $u\in\{a_1,a_2,a_3\}$ and $x,y\in\{c_1,c_2,c_3\}$. But now we have a contradiction to the fact that $(u,x)\in E(G)$. Thus $c\notin\{u,x,y,v\}$. Symmetrically, we also get $d\notin\{u,x,y,v\}$. If $a_i,a_j\in\{u,x,y,v\}$ for some distinct $i,j\in\{1,2,3\}$, then $\{a_i,a_j\}=\{x,y\}$, since $a_i$ and $a_j$ are true twins. But then in $\phi$, there is no common neighbour of $a_i$ and $a_j$ that is not a true twin of theirs to the left of $x$, so $u$ has to be a true twin of $x$ and $y$, which is a contradiction. Therefore, at most one among $a_1,a_2,a_3$ can be present in $\{u,x,y,v\}$, so we shall assume without loss of generality that $a_2,a_3\notin\{u,x,y,v\}$. Using similar arguments, we shall assume without loss of generality that $b_2,b_3\notin\{u,x,y,v\}$, $c_2,c_3\notin\{u,x,y,v\}$ and $d_2,d_3\notin\{u,x,y,v\}$. 

First, suppose $c_1\in\{u,x,y,v\}$, then since $c_1$ has no neighbours to the left of it in $\phi$, we have $c_1=u$, and $v \in \{b,d\}$ as they are the only non-neighbours of $c_1$ to the right of it. This is a contradiction to our earlier observation that $b,d\notin\{u,x,y,v\}$. Thus $c_1 \not\in \{u,x,y,v\}$. Symmetrically, $d_1 \not\in \{u,x,y,v\}$.     

Now suppose that $a_1\in\{u,x,y,v\}$. Then since $a_1$ has no neighbours to its left, we have $u=a_1$. As $(u,v)\notin E(G')$, we now have $v\in F$. 
Also, since $x$ and $y$ are neighbours of $u$, we have $x,y\in V(G)\cup\{b_1\}$. Since $x$ and $y$ occur before $v$ in the ordering, we further have that $x,y\in V(G)$. But now we have the contradiction that in $\phi$, the vertex $v\in F$ does not lie between its two neighbours $x,y\in V(G)$. Thus $a_1\notin\{u,x,y,v\}$. Symmetrically, we have $b_1\notin\{u,x,y,v\}$. 
Thus no vertex from $A \cup B \cup C \cup D$ can be in $\{u,x,y,v\}$, which implies that $\{u,x,y,v\}\subseteq M$. But this means that $u,x,y,v$ form an ordered $K_4^-$ in $\sigma$, which a contradiction to the previous claim.
\end{proof}
\bigskip
	
\noindent This proves that if $G$ is 3-colourable, then $G'$ has a 3-C-E ordering.
\end{proof}

\bigskip

Lemma~\ref{lem:G'toG} and Lemma~\ref{lem:GtoG'} prove the correctness of the reduction and thus we have the following theorem.

\begin{theorem}
	{\ProbFind} is {\sc NP}-hard for $k=3$.
\end{theorem}

\section{Conclusion} \label{section:conclusion}
We have shown that the problem of determining whether a given graph is a $k$-C-E is NP-hard for each $k \geq 3$ and {\sc co-NP}-hard for general $k$. Finding a maximum clique in a $k$-C-E graph on $n$ vertices is known to have an $n^{O(k)}$ algorithm when a $k$-clique-extendible ordering is given, which we prove to be optimal. 
It is also an open problem mentioned before~\cite{spinrad} whether we can find a maximum clique in a $k$-C-E graph in polynomial time for a fixed $k$, if given only the adjacency matrix of the graph. 
Finally, it would be interesting to know polynomial time solvable problems in $k$-C-E graphs, even for $k=3$.
As triangle free graphs and diamond-free graphs are $3$-C-E graphs, we know that the independent set problem and the colouring problem are NP-hard in these classes of graphs.

It would also be interesting to study whether these graphs can be recognised approximately. There are two suitable notions for approximation. One is the following: An algorithm is said to be an $\alpha $-factor approximation (for $\alpha \ge 1$) if, given a graph $G$ and integer $k$, it either outputs a $(\alpha k)$-C-E ordering or concludes that no $k$-C-E ordering exists for $G$. The second notion is the following: An algorithm is said to be a $\alpha $-factor approximation (for $\alpha \le 1$) if, given a graph $G$ and integer $k$, outputs an ordering $\phi $ such that at most $\alpha $ fraction of the induced $K_{k+1}^-$ in the graph are ordered in $\phi $. Note that solving this problem for $\alpha =0$ is equivalent to solving {\ProbFind}.

For the second notion of approximation, there is an easy $(\frac{2}{k(k+1)})$-factor approximation. Simply output a random ordering of the vertices of $G$. The probability that any given induced $K_{k+1}^-$ is ordered is $\frac{2}{k(k+1)}$. Thus by linearity of expectation, a $\frac{2}{k(k+1)}$ fraction of all the induced $K_{k+1}^-$ in $G$ will be ordered.  
%\textbf{\textcolor{red}{TODO}}

\bibliography{refs}

\end{document}